\documentclass[a4paper,11pt]{article}
\usepackage{graphicx,epsfig}
\usepackage{fancyhdr,fancybox}
\usepackage{indentfirst}
\usepackage{verbatim}
\usepackage[numbers,sort&compress]{natbib}
\usepackage{geometry}
\usepackage{extarrows} 
\usepackage{color}
\usepackage[small]{caption2}
\usepackage{microtype}
\DisableLigatures[f]{encoding = *, family = *}
\usepackage[none]{hyphenat}

\usepackage{times}     

\usepackage{amssymb}                 
\usepackage{amsthm}
\usepackage{mathrsfs}                
\usepackage{bbm}
\usepackage{hyperref}  

\newcommand{\paperfont}{\fontsize{11pt}{1.2\baselineskip}\selectfont}
\pagebreak[4]
\begin{document}

\theoremstyle{definition}
\makeatletter
\thm@headfont{\bf}
\makeatother
\newtheorem{theorem}{Theorem}[section]
\newtheorem{definition}[theorem]{Definition}
\newtheorem{lemma}[theorem]{Lemma}
\newtheorem{proposition}[theorem]{Proposition}
\newtheorem{corollary}[theorem]{Corollary}
\newtheorem{remark}[theorem]{Remark}
\newtheorem{example}[theorem]{Example}
\newtheorem{assumption}[theorem]{Assumption}

\lhead{}
\rhead{}
\lfoot{}
\rfoot{}

\renewcommand{\refname}{References}
\renewcommand{\figurename}{Figure}
\renewcommand{\tablename}{Table}
\renewcommand{\proofname}{Proof}

\newcommand{\diag}{\mathrm{diag}}
\newcommand{\tr}{\mathrm{tr}}
\newcommand{\re}{\mathrm{Re}}
\newcommand{\one}{\mathbbm{1}}
\newcommand{\loc}{\textrm{loc}}

\newcommand{\Pnum}{\mathbb{P}}
\newcommand{\Enum}{\mathbb{E}}
\newcommand{\Rnum}{\mathbb{R}}
\newcommand{\Cnum}{\mathbb{C}}
\newcommand{\Znum}{\mathbb{Z}}
\newcommand{\Nnum}{\mathbb{N}}
\newcommand{\abs}[1]{\left\vert#1\right\vert}
\newcommand{\set}[1]{\left\{#1\right\}}
\newcommand{\norm}[1]{\left\Vert#1\right\Vert}
\newcommand{\innp}[1]{\langle {#1}\rangle}
\newcommand{\style}{\setlength{\itemsep}{1pt}\setlength{\parsep}{1pt}\setlength{\parskip}{1pt}}

\title{\textbf{Martingale structure for general thermodynamic functionals of diffusion processes under second-order averaging}}
\author{Hao Ge$^{1,2,*}$,\;\;\;Chen Jia$^{3,*}$,\;\;\;Xiao Jin$^{4}$ \\
\footnotesize $^1$ Beijing International Center for Mathematical Research, Peking University, Beijing 100871, China\\
\footnotesize $^2$ Biomedical Pioneering Innovation Center, Peking University, Beijing 100871, China\\
\footnotesize $^3$ Division of Applied and Computational Mathematics, Beijing Computational Science Research Center, Beijing 100193, China. \\
\footnotesize $^4$ School of Mathematical Sciences, Peking University, Beijing 100871, China\\
\footnotesize $^*$ Correspondence: haoge@pku.edu.cn (H. G.), chenjia@csrc.ac.cn (C. J.)}
\date{}                              
\maketitle                           
\thispagestyle{empty}                

\paperfont

\begin{abstract}
Novel hidden thermodynamic structures have recently been uncovered during the investigation of nonequilibrium thermodynamics for multiscale stochastic processes. Here we reveal the martingale structure for a general thermodynamic functional of inhomogeneous singularly perturbed diffusion processes under second-order averaging, where a general thermodynamic functional is defined as the logarithmic Radon-Nykodim derivative between the laws of the original process and a comparable process (forward case) or its time reversal (backward case). In the forward case, we prove that the regular and anomalous parts of a thermodynamic functional are orthogonal martingales. In the backward case, while the regular part may not be a martingale, we prove that the anomalous part is still a martingale. With the aid of the martingale structure, we prove the integral fluctuation theorem satisfied by the regular and anomalous parts of a general thermodynamic functional. Further extensions and applications to stochastic thermodynamics are also discussed, including the martingale structure and fluctuation theorems for the regular and anomalous parts of entropy production and housekeeping heat in the absence or presence of odd variables. \\

\noindent 
\textbf{AMS Subject Classifications}: 60J60, 58J65, 60G44, 46N55, 82C05, 82C35 \\
\textbf{Keywords}: singularly perturbed diffusion, model simplification, multiscale, stochastic thermodynamics, entropy production, housekeeping heat, anomalous contribution, fluctuation theorem
\end{abstract}

\section{Introduction}
Over the past two decades, significant progress has been made in stochastic thermodynamics \cite{jarzynski2011equalities, seifert2012stochastic, van2015ensemble}, which has grown to become one of the most influential branches of statistical physics. In this field, an irreversible thermodynamic system is usually modeled by a nonstationary and nonsymmetric Markov process, where the breaking of stationarity describes Boltzmann's irreversibility and the breaking of symmetry describes Prigogine's irreversibility \cite{esposito2010letter, esposito2010three, ge2010physical}. Along this line, an equilibrium state is defined as a stationary and symmetric Markov process and the deviation from equilibrium is usually quantified by the concept of entropy production \cite{jiang2004mathematical, zhang2012stochastic}. Entropy production can be decomposed into two parts: the adiabatic part which is also referred to as housekeeping heat and the nonadiabatic part which is also known as free energy dissipation \cite{esposito2010letter, esposito2010three, ge2010physical}. The former characterizes the deviation from symmetry and the latter characterizes the deviation from stationarity \cite{hong2016novel}.

Entropy production, as well as its adiabatic and nonadiabatic parts, can also be defined along a single stochastic trajectory. These thermodynamic quantities are all non-additional functionals of the stochastic trajectory of the underlying Markovian model. One of the major breakthroughs in stochastic thermodynamics is the finding that a broad class of thermodynamic functionals such as entropy production, housekeeping heat, free energy dissipation, and cycle fluxes \cite{andrieux2007fluctuation, jia2016cycle, ge2017cycle} satisfy various types of fluctuation theorems \cite{evans1993probability, gallavotti1995dynamical,jarzynski1997nonequilibrium, sekimoto2010stochastic, kurchan1998fluctuation, crooks1999entropy, searles1999fluctuation, lebowitz1999gallavotti, hatano2001steady, seifert2005entropy, ritort2008nonequilibrium}, which provide nontrivial generalizations of the second law of thermodynamics in terms of equalities rather than inequalities. These fluctuation theorems, which hold generally for states far from equilibrium, have garnered considerable attention from physicists, chemists, biologists, and mathematicians. Thus far, most of these fluctuation theorems have been numerically tested \cite{bonetto1997chaotic, hendrix2001fast, van2003extension, giuliani2005fluctuation, williams2006numerical} or experimentally validated \cite{ciliberto1998experimental, liphardt2002equilibrium, garnier2005nonequilibrium, collin2005verification, imparato2007work}.

Mathematically, there is a universal approach for the investigation of fluctuation theorems. It has been shown that most thermodynamic functionals have a unified mathematical representation in terms of the logarithmic Radon-Nykodim derivative between the laws of the original process and some comparable process \cite{maes2003origin, chernyak2006path, chetrite2008fluctuation}:
\begin{equation}\label{RN}
A_t = \log\frac{dP_t}{dQ_t}(X_\cdot),
\end{equation}
where $A_t$ is a thermodynamic functional, $P_t$ is the law of the original process $X = (X_s)_{0\leq s\leq t}$, and $Q_t$ is the law of the comparable process $Y = (Y_s)_{0\leq s\leq t}$. If the thermodynamic functional of interest has the representation \eqref{RN}, then the following integral fluctuation theorem naturally follows:
\begin{equation*}
\Enum e^{-A_t} = 1.
\end{equation*}
It then follows from Jensen's inequality that $\Enum A_t\geq 0$, which is just the conventional inequality in the second law of thermodynamics.

In recent years, another striking fact in stochastic thermodynamics beyond fluctuation theorems is that, under mild conditions, a wide class of thermodynamic functionals such as entropy production and housekeeping heat are actually martingales \cite{chetrite2011two}. This not only reveals a deep connection between statistical physics and probability theory, but also enables us to apply the powerful tool of martingale theory to derive some novel statistical properties of such thermodynamic functionals, including universal equalities and inequalities for stopping times, passage statistics, and extrema \cite{neri2017statistics, pigolotti2017generic, chetrite2019martingale}. An experimental validation of the infimum law of entropy production with a double electronic dot has been performed recently \cite{singh2017records}.

In stochastic thermodynamics, one of most important mathematical models of a molecular system is the diffusion process, which generalizes the classical Langevin equation describing the stochastic movement of a mesoscopic particle in a fluid due to collisions with the fluid molecules. Diffusion processes have also been widely applied to model various stochastic systems in chemistry, biology, meteorology, and other disciplines \cite{pavliotis2014stochastic}. In practice, a complex molecular system (such as biochemical reaction networks in living cells) often has multiple different time scales spanning many orders of magnitude \cite{jia2017simplification}. For efficient modeling, one usually averages out rapidly varying components and only focuses on slowly varying components. This procedure is called averaging \cite{pavliotis2008multiscale}, coarse graining \cite{pigolotti2008coarse}, or model simplification \cite{jia2016reduction, jia2016simplification}. Mathematically, many previous papers focused on a two-timescale diffusion process whose generator has the singularly perturbed form of
\begin{equation*}
L^\epsilon = \epsilon^{-1}L_0+L_1,
\end{equation*}
where $\epsilon>0$ is a small parameter, $L_0$ is the generator governing the fast components, and $L_1$ is the generator governing the slow components \cite{khasminskii1996transition, khasminskii2004averaging, khasminskii2011stochastic}. The model simplification of such systems is called first-order averaging. In statistical physics, however, many important molecular systems such as the underdamped Langevin equation are actually singularly perturbed diffusion processes with a more complicated generator \cite{spinney2012entropy, lee2013fluctuation, ge2014time}
\begin{equation*}
L^\epsilon = \epsilon^{-2}L_0+\epsilon^{-1}L_1+L_2.
\end{equation*}
The model simplification of such systems is called second-order averaging which has also been studied extensively in probability theory \cite{pardoux2001poisson, pardoux2003poisson, khasminskii2005limit}.

Recently, it has been found that while model simplification such as first-order and second-order averaging retains the dynamic information of the original process, it loses some crucial thermodynamic information as a tradeoff \cite{esposito2012stochastic, kawaguchi2013fluctuation, puglisi2010entropy,rahav2007fluctuation, santillan2011irreversible, nakayama2015invariance, ford2015stochastic, celani2012anomalous, bo2014entropy, lan2015stochastic, jia2016model, marino2016entropy, bo2017multiple, birrell2018entropy}. In particular, under model simplification, it has been shown that entropy production and housekeeping heat usually decrease in the average sense \cite{esposito2012stochastic, jia2016model}. The decreasing part emerging from the elimination of fast components is called the anomalous part of the corresponding thermodynamic functional. In addition, there have been some results showing that the anomalous part of entropy production and housekeeping heat satisfies fluctuation theorems \cite{kawaguchi2013fluctuation, bo2014entropy}. However, it is unclear whether the anomalous parts of such thermodynamic functionals are martingales, and whether similar results also hold for a general thermodynamic functional defined by the representation \eqref{RN}.

In this paper, we investigate the martingale structure for a general thermodynamic functional with the representation \eqref{RN} of inhomogeneous singularly perturbed diffusion processes under second-order averaging. Different thermodynamic functionals correspond to different comparable processes. For most thermodynamic functionals studied in the literature, the comparable process is either a diffusion process having the same diffusion coefficients as the original process or the time reversal of such a diffusion process under the reversed protocol \cite{chetrite2008fluctuation}. To distinguish between these two scenarios, we refer to the former as the forward case and the latter as the backward case. Here we consider the martingale structure of a thermodynamic functional under the most general setup covering these two scenarios, leaving most known physical examples as special cases.

The main results of the present paper are stated as follows. Let $F^\epsilon_t$ be a general forward or backward thermodynamic functional for an inhomogeneous singularly perturbed diffusion process and let $F^{(1)}_t$ be the corresponding thermodynamic functional for the reduced process under second-order averaging. We prove that under mild conditions that are usually satisfied by real physical systems, the thermodynamic functional $F^\epsilon_t$ has a weak limit $F_t$ as $\epsilon\rightarrow 0$. The limit functional $F_t$ is in general not equal to the thermodynamic functional $F^{(1)}_t$ of the reduced process and thus can be decomposed into two parts:
\begin{equation*}
F_t = F^{(1)}_t+F^{(2)}_t,
\end{equation*}
where $F^{(1)}_t$ and $F^{(2)}_t$ are called the regular and anomalous parts of the limit functional, respectively.

In this paper, we prove that in the forward scenario, the thermodynamic functional $e^{-F^\epsilon_t}$, the limit functional $e^{-F_t}$, its regular part $e^{-F^{(1)}_t}$, and its anomalous part $e^{-F^{(2)}_t}$ are all martingales. Moreover, the regular and anomalous parts are actually orthogonal martingales. This provides an orthogonal martingale decomposition of the forward limit functional. In the backward scenario, the thermodynamic functional $e^{-F^\epsilon_t}$, the limit functional $e^{-F_t}$, and its regular part $e^{-F^{(1)}_t}$ may not be martingales due to a non-vanishing boundary term depending on the initial and final distributions. However, we prove that the anomalous part $e^{-F^{(2)}_t}$ is still a martingale. These results reveal a highly nontrivial relationship between model simplification and martingale. In particular, we show that the anomalous part of a general thermodynamic functional is always a martingale. This martingale structure enables us to prove the following integral fluctuation theorem:
\begin{equation*}
\Enum e^{-F^{(2)}_{\tau}} = 1,
\end{equation*}
where $\tau$ is an arbitrary bounded stopping time. Therefore, as a byproduct, we prove that the anomalous part of a general thermodynamic functional satisfies the integral fluctuation theorem.

The structure of this paper is organized as follows. In Section 2, we introduce our model and the fundamental convergence theorem for singularly perturbed diffusion processes under second-order averaging. In Sections 3 and 4, we state our main results about the martingale structure and fluctuation theorems for general forward and backward thermodynamic functionals. In Sections 5 and 6, we apply our general theorems to stochastic thermodynamics. In particular, we establish the martingale structure for entropy production and housekeeping heat in the absence and presence of odd variables. Sections 7 and 8 are devoted to the proof of our main theorems.

\section{Formulation}
We consider a thermodynamic process modeled by a pair of inhomogeneous singularly perturbed diffusions $(X^\epsilon,Y^\epsilon) = (X^\epsilon_t,Y^\epsilon_t)_{0\leq t\leq T}$ with a small parameter $\epsilon>0$ and a fixed ending time $T>0$, where $X^\epsilon$ is a diffusion with values in $\Rnum^m$ and $Y^\epsilon$ is a diffusion with values in $\Rnum^n$. There are weak and strong interactions between these processes so that $X^\epsilon$ is slowly varying and $Y^\epsilon$ is fast varying. The system of diffusions $(X^\epsilon,Y^\epsilon)$ solves the stochastic differential equation (SDE)
\begin{equation}\label{secondorder}\left\{
\begin{split}
dX^{\epsilon}_t
&= [b(X^\epsilon_t,Y^\epsilon_t,t)+\epsilon^{-1}f(X^\epsilon_t,Y^\epsilon_t,t)]dt
+\sigma(X^\epsilon_t,Y^\epsilon_t,t)dW_t,\\
dY^{\epsilon}_t
&= [\epsilon^{-1}g(X^\epsilon_t,Y^\epsilon_t,t)+\epsilon^{-2}c(X^\epsilon_t,Y^\epsilon_t,t)]dt
+\epsilon^{-1}\eta(X^\epsilon_t,Y^\epsilon_t,t)dW_t,
\end{split}\right.
\end{equation}
where $W = (W_t)_{t\geq 0}$ is a $p$-dimensional standard Brownian motion defined on a filtered probability space $(\Omega,\mathcal{F},\{\mathscr{F}_t\},\Pnum)$, $b,f:\Rnum^m\times\Rnum^n\times\Rnum^+\rightarrow\Rnum^m$ and
$c,g:\Rnum^m\times\Rnum^n\times\Rnum^+\rightarrow\Rnum^n$ are vector-valued functions, and
$\sigma:\Rnum^m\times\Rnum^n\times\Rnum^+\rightarrow\Rnum^m\times\Rnum^p$
and $\eta:\Rnum^m\times\Rnum^n\times\Rnum^+\rightarrow\Rnum^n\times\Rnum^p$ are matrix-valued functions. Then $(X^\epsilon,Y^\epsilon)$ is an inhomogeneous diffusion process on $\Rnum^{n+m}$ with drift
\begin{equation*}
B(x,y,t) = \begin{pmatrix}
b(x,y,t)+\epsilon^{-1}f(x,y,t)\\
\epsilon^{-1}g(x,y,t)+\epsilon^{-2}c(x,y,t) \end{pmatrix}
\end{equation*}
and diffusion matrix
\begin{equation*}
D(x,y,t) = \begin{pmatrix}
\sigma(x,y,t)\\
\epsilon^{-1}\eta(x,y,t) \end{pmatrix}\begin{pmatrix}
\sigma(x,y,t)\\
\epsilon^{-1}\eta(x,y,t) \end{pmatrix}' = \begin{pmatrix}
a(x,y,t) & \epsilon^{-1}h(x,y,t) \\
\epsilon^{-1}h(x,y,t)' & \epsilon^{-2}\alpha(x,y,t)\end{pmatrix}.
\end{equation*}
where
\begin{equation*}
a = \sigma\sigma',\;\;\;h = \sigma\eta',\;\;\;\alpha = \eta\eta'.
\end{equation*}
Following previous papers \cite{pardoux2001poisson, pardoux2003poisson, khasminskii2005limit}, we assume that $(X^\epsilon,Y^\epsilon)$ starts from the same initial distribution for each $\epsilon>0$. Recall that the process is called homogenous if the drift and diffusion coefficients are independent of time $t$. Otherwise, it is called inhomogeneous. Inhomogeneous diffusion processes are particularly important in statistical physics because they can be used to describe the evolution of a molecular system controlled by a time-varying protocol.

Throughout this paper, we use \textbf{Einstein's summation convention}: if the same index appears twice in any term, once as an upper index and once as a lower index, that term is understood to be summed over all possible values of that index. Then the generator of $(X^\epsilon,Y^\epsilon)$ has the following second-order singularly perturbed form:
\begin{equation}\label{generator}
L^\epsilon = \epsilon^{-2}L_0+\epsilon^{-1}L_1+L_2,
\end{equation}
where
\begin{gather*}
L_0 = c^i\partial_{y_i}+\frac{1}{2}\alpha^{ij}\partial_{y_iy_j},\\
L_1 = f^i\partial_{x_i}+g^i\partial_{y_i}+h^{ij}\partial_{x_iy_j},\\
L_2 = b^i\partial_{x_i}+\frac{1}{2}a^{ij}\partial_{x_ix_j}.
\end{gather*}
In the special case of $f = 0$, $g = 0$, and $h = 0$, the SDE \eqref{secondorder} reduces to
\begin{equation*}\left\{
\begin{split}
dX^{\epsilon}_t &= b(X^\epsilon_t,Y^\epsilon_t,t)dt+\sigma(X^\epsilon_t,Y^\epsilon_t,t)dW_t,\\
dY^{\epsilon}_t &= \epsilon^{-2}c(X^\epsilon_t,Y^\epsilon_t,t)dt
+\epsilon^{-1}\eta(X^\epsilon_t,Y^\epsilon_t,t)dW_t.
\end{split}\right.
\end{equation*}
In this case, the generator of $(X^\epsilon,Y^\epsilon)$ has the following first-order singularly perturbed form:
\begin{equation*}
L^\epsilon = \epsilon^{-2}L_0+L_2.
\end{equation*}
Compared with the first-order case, second-order singularly perturbed diffusions can deal with more complicated drift terms, as well as correlated noise terms between fast and slow components. In addition, most physics papers focus on Langevin equations with the diffusion matrix $D$ being diagonal. Our setting can deal with the non-diagonal coupling case.

Treating $x\in\Rnum^m$ as a parameter, we next consider a family of $\Rnum^n$-valued inhomogeneous diffusion processes $(Y^x)_{x\in\Rnum^m}$ with generator $L_0$, which solves the SDE
\begin{equation*}
dY^x_t = c(x,Y^x_t,t)dt+\eta(x,Y^x_t,t)dW_t.
\end{equation*}
In this paper, we assume that the following regularity conditions are satisfied:
\begin{itemize}\style
\item[(a)] The functions $b,f,c,g,\sigma,\eta$ are bounded and smooth and all their partial derivatives are bounded.
\item[(b)] The diffusion matrix $\alpha$ satisfies the uniformly elliptic condition, i.e. there exists $\lambda>0$ such that
    \begin{equation*}
    \xi^T\alpha(x,y,t)\xi \geq \lambda|\xi|^2,\;\;\;\textrm{for any}\;\xi\in\Rnum^n.
    \end{equation*}
\item[(c)] For any $x\in\Rnum^m$ and $t\geq 0$, the process $Y^x$ has a pseudo-stationary density $\rho = \rho(x,y,t)$ that is the solution of
    \begin{equation}\label{rho}
    L_0^*\rho = 0,\;\;\;\int_{\Rnum^n}\rho(x,y,t)dy = 1,
    \end{equation}
    where $L_0^*$ is the adjoint operator of $L_0$ with respect to the Lebesgue measure.
\item[(d)] For any $x\in\Rnum^m$ and $t\geq 0$, the function $f$ satisfies the centering condition
    \begin{equation}\label{centering}
    \int_{\Rnum^n}f(x,y,t)\rho(x,y,t)dy = 0.
    \end{equation}
\end{itemize}

In fact, the regularity condition (a) can be weakened by imposing several Lipschitz, H\"{o}lder, and growth conditions \cite{pardoux2001poisson, pardoux2003poisson}. Since our main concern here is to understand the martingale structure of a general thermodynamic functional under model simplification, we do not pursue the weakest regularity conditions. Under the regularity conditions (a)-(c), it is a classical result that the pseudo-stationary density of $Y^x$ must be unique, smooth, and positive for any $x\in\Rnum^m$ and $t\geq 0$ \cite{jiang2004mathematical}. Moreover, under all the regularity conditions, it is also known that there exists a unique smooth solution $\phi$ to the Poisson equation \cite{pardoux2001poisson, pardoux2003poisson}
\begin{equation*}\left\{
\begin{split}
&-L_0\phi(x,y,t) = f(x,y,t),\\
&\int_{\Rnum^n}\phi(x,y,t)\rho(x,y,t)dy = 0.
\end{split}\right.
\end{equation*}
In fact, the asymptotic behavior of the slow component $X^\epsilon$ under second-order averaging has been well established when $g = 0$ \cite{pardoux2001poisson, pardoux2003poisson}. The asymptotic behavior of $X^\epsilon$ in the general case of $g\neq 0$ can be easily obtained using similar methods \cite{hu2017hypoelliptic}, as stated in the following lemma.

\begin{lemma}\label{convergence}
Let $C([0,T],\Rnum^m)$ be the space of continuous functions $f:[0,T]\rightarrow\Rnum^m$. If $(X^\epsilon,Y^\epsilon)$ starts from the same initial distribution for each $\epsilon>0$, then $X^\epsilon$ converges weakly in $C([0,T],\Rnum^m)$ to another diffusion process $X = (X_t)_{0\leq t\leq T}$ as $\epsilon\rightarrow 0$. Moreover, the generator of $X$ is given by
\begin{equation*}
L = w^i(x,t)\partial_{x_i}+\frac{1}{2}A^{ij}(x,t)\partial_{x_ix_j},
\end{equation*}
where
\begin{gather*}
w^i(x,t) = \int_{\Rnum^n}(b^i+\partial_{x_j}\phi^if^j+\partial_{y_j}\phi^ig^j
+\partial_{x_jy_k}\phi^ih^{jk})(x,y,t)\rho(x,y,t)dy,\\
A^{ij}(x,t) = \int_{\Rnum^n}(a^{ij}+\phi^if^j+\phi^jf^i+\partial_{y_k}\phi^ih^{jk}
+\partial_{y_k}\phi^jh^{ik})(x,y,t)\rho(x,y,t)dy.
\end{gather*}
\end{lemma}

\section{Martingale structure for forward thermodynamic functionals}
In this section, we investigate the martingale structure for forward thermodynamic functionals of inhomogeneous singularly perturbed diffusion processes under second-order averaging. Before doing this, recall that for any $\Rnum^m$-valued diffusion process $Z = (Z_t)_{0\leq t\leq T}$, the distribution (or the law) of $Z$ is a probability measure $P^Z_T$ on the path space $C([0,T],\Rnum^m)$ defined as
\begin{equation*}
P^Z_T(A) = \Pnum(Z_\cdot\in A),
\end{equation*}
where $A\subset C([0,T],\Rnum^m)$ is an arbitrary Borel set.

In this paper, we define the forward thermodynamic functional for a diffusion process as the logarithmic Radon-Nikodym derivative between the laws of the original process and another one, which is called the comparable process. Specifically, let $(\hat{X}^\epsilon,\hat{Y}^\epsilon) = (\hat{X}^\epsilon_t,\hat{Y}^\epsilon_t)_{0\leq t\leq T}$ be a comparable diffusion process solving the SDE
\begin{equation}\label{comparable1}\left\{
\begin{split}
d\hat{X}^\epsilon_t &= [\hat{b}(\hat{X}^\epsilon_t,\hat{Y}^\epsilon_t,t)+
\epsilon^{-1}f(\hat{X}^\epsilon_t,\hat{Y}^\epsilon_t,t)]dt
+\sigma(\hat{X}^\epsilon_t,\hat{Y}^\epsilon_t,t)dW_t,\\
d\hat{Y}^\epsilon_t &= [\epsilon^{-1}\hat{g}(\hat{X}^\epsilon_t,\hat{Y}^\epsilon_t,t)
+\epsilon^{-2}c(\hat{X}^\epsilon_t,\hat{Y}^\epsilon_t,t)]dt
+\epsilon^{-1}\eta(\hat{X}^\epsilon_t,\hat{Y}^\epsilon_t,t)dW_t,
\end{split}\right.
\end{equation}
where $\hat{b}:\Rnum^m\times\Rnum^n\times\Rnum^+\rightarrow\Rnum^m$ and
$\hat{g}:\Rnum^m\times\Rnum^n\times\Rnum^+\rightarrow\Rnum^n$. The initial distribution of $(\hat{X}^\epsilon,\hat{Y}^\epsilon)$ is chosen to be the same as that of $(X^\epsilon,Y^\epsilon)$. We also assume that the comparable process satisfies the regularity conditions (a)-(d). The forward thermodynamic functional for the original process is defined rigorously as follows.

\begin{definition}
Let $P^\epsilon_T$ be the law of the original process $(X^\epsilon,Y^\epsilon)$ and let $\hat{P}^\epsilon_T$ be the law of the comparable process $(\hat{X}^\epsilon,\hat{Y}^\epsilon)$. Then the forward thermodynamic functional of $(X^\epsilon,Y^\epsilon)$ with respect to $(\hat{X}^\epsilon,\hat{Y}^\epsilon)$ is defined as
\begin{equation*}
F^\epsilon_T = \log\frac{dP^\epsilon_T}{d\hat{P}^\epsilon_T}(X^\epsilon_\cdot,Y^\epsilon_\cdot),
\end{equation*}
if $P^\epsilon_T$ and $\hat{P}^\epsilon_T$ are absolutely continuous with respect to each other.
\end{definition}

Since the original and comparable processes are both singularly perturbed diffusions, it follows from Lemma \ref{convergence} that $X^\epsilon$ converges weakly in $C([0,T],\Rnum^m)$ to another diffusion $X = (X_t)_{0\leq t\leq T}$ with generator
\begin{equation*}
L = w^i(x,t)\partial_i+\frac{1}{2}A^{ij}(x,t)\partial_{ij},
\end{equation*}
where $w^i$ and $A^{ij}$ are defined in Lemma \ref{convergence}, and $\hat{X}^\epsilon$ converges weakly in $C([0,T],\Rnum^m)$ to another diffusion $\hat{X} = (\hat{X}_t)_{0\leq t\leq T}$ with generator
\begin{equation*}
\hat{L} = \hat{w}^i(x,t)\partial_{x_i}+\frac{1}{2}A^{ij}(x,t)\partial_{x_ix_j},
\end{equation*}
where
\begin{equation*}
\hat{w}^i(x,t) = \int_{\Rnum^n}(\hat{b}^i+\partial_{x_j}\phi^if^j+\partial_{y_j}\phi^i\hat{g}^j
+\partial_{x_jy_k}\phi^ih^{jk})(x,y,t)\rho(x,y,t)dy.
\end{equation*}
It is important to notice that $X$ and $\hat{X}$ share the same diffusion matrix $A = (A^{ij})$. The forward thermodynamic function for the averaged process is defined rigorously as follows.

\begin{definition}
Let $P_T$ be the law of $X$ and let $\hat{P}_T$ be the law of $\hat{X}$. Then the forward thermodynamic function of $X$ with respect to $\hat{X}$ is defined as
\begin{equation*}
F^{(1)}_T = \log\frac{dP_T}{d\hat{P}_T}(X_\cdot),
\end{equation*}
if $P_T$ and $\hat{P}_T$ are absolutely continuous with respect to each other.
\end{definition}

\begin{remark}
If we let the ending time $T$ change, then both $F^\epsilon_t$ and $F^{(1)}_t$ can be viewed as stochastic processes with $T$ being replaced by an arbitrary time $t$.
\end{remark}

Since $X^\epsilon\Rightarrow X$ in $C([0,T],\Rnum^m)$, one may guess that the forward thermodynamic functional $F^\epsilon_t$ of the original process should converge weakly to the forward thermodynamic functional $F^{(1)}_t$ of the averaged process. However, recent studies in stochastic thermodynamics indicate that this is generally not true \cite{esposito2010three}. In fact, $F^\epsilon_t$ indeed converges weakly to a limit functional $F_t$, as stated in the following theorem. However, $F_t$ is in general not equal to $F^{(1)}_t$, which means that $F_t$ can be decomposed into two parts:
\begin{equation*}
F_t = F^{(1)}_t+F^{(2)}_t.
\end{equation*}
Recall that two continuous martingales $M = (M_t)_{0\leq t\leq T}$ and $N = (N_t)_{0\leq t\leq T}$ are called orthogonal if their quadratic variation process $[M,N]$ vanishes.

\begin{theorem}\label{forward}
The following conclusions hold:
\begin{itemize}\style
\item[(a)] For any $t\geq 0$, both $F^\epsilon_t$ and $F^{(1)}_t$ are well defined.
\item[(b)] The pair of processes $(X^\epsilon_t,F^\epsilon_t)_{0\leq t\leq T}$ converges weakly to another pair of processes $(X_t,F_t)_{0\leq t\leq T}$ in $C([0,T],\Rnum^{m+1})$.
\item[(c)] Both $e^{-F^\epsilon_t}$ and $e^{-F_t}$ are martingales.
\item[(d)] If we decompose the limit functional into $F_t = F^{(1)}_t+F^{(2)}_t$, then $e^{-F^{(1)}_t}$ and $e^{-F^{(2)}_t}$ are orthogonal martingales.
\end{itemize}
\end{theorem}

\begin{proof}
The proof of this theorem will be given in Section \ref{forwardproof}.
\end{proof}

In nonequilibrium statistical physics, $F^{(1)}_t$ and $F^{(2)}_t$ are referred to as the regular and anomalous parts of the limit functional $F_t$, respectively. In general, the anomalous part does not vanish. The above theorem shows that $e^{-F_t}$ is a martingale that can be decomposed into the product of two orthogonal martingales:
\begin{equation*}
e^{-F_t} = e^{-F^{(1)}_t}e^{-F^{(2)}_t}.
\end{equation*}
This reveals a nontrivial orthogonal decomposition for forward thermodynamic functionals.

\begin{remark}\label{infinitesimal}
From of the proof in Section \ref{forwardproof}, it is not hard to verify that the results of the above theorem also hold if the functions $\hat{b}$ and $\hat{g}$ in the comparable process are replaced by $\hat{b}+o(1)$ and $\hat{g}+o(1)$, respectively. In other words, the proof of our main results is not affected by these $o(1)$ terms.
\end{remark}

The next corollary, which follows directly from the above theorem and the optional sampling theorem for martingales, is particularly important in nonequilibrium statistical physics.

\begin{corollary}[\textbf{Integral fluctuation theorem}]
The forward functionals $F^\epsilon_t$, $F_t$, $F^{(1)}_t$, and $F^{(2)}_t$ satisfy the following integral fluctuation theorem: for any bounded stopping time $\tau$ with respect to $\{\mathscr{F}_t\}$, we have
\begin{equation*}
\Enum e^{-F^\epsilon_\tau} = \Enum e^{-F_\tau} = \Enum e^{-F^{(1)}_\tau} = \Enum e^{-F^{(2)}_\tau} = 1.
\end{equation*}
\end{corollary}

\section{Martingale structure for backward thermodynamic functionals}
In this section, we investigate the martingale structure for backward thermodynamic functionals of inhomogeneous singularly perturbed diffusion processes under second-order averaging. To avoid the initial layer effect, we assume that the original process $(X^\epsilon,Y^\epsilon)$ starts at $t = -1$. While the original process starts at $t = -1$, we still focus on the its dynamic behavior during $t\in[0,T]$. The reason why we make this assumption is to guarantee the convergence of backward thermodynamic functionals under second-order averaging. Under this assumption, there is no need to impose additional requirements on the initial distribution of the original process at time $0$.

In this paper, we define the backward thermodynamic functional for a diffusion process as the logarithmic Radon-Nikodym derivative between the laws of the original process and the time reversal of a comparable process. Specifically, let $(\tilde{X}^{R,\epsilon},\tilde{Y}^{R,\epsilon}) = (\tilde{X}^{R,\epsilon}_t,\tilde{Y}^{R,\epsilon}_t)_{0\leq t\leq T}$ be a comparable diffusion process solving the SDE
\begin{equation}\label{comparable2}\left\{
\begin{split}
d\tilde{X}^{R,\epsilon}_t &= (\tilde{b}+\epsilon^{-1}\tilde{f})
(\tilde{X}^{R,\epsilon}_t,\tilde{Y}^{R,\epsilon}_t,T-t)dt
+\sigma(\tilde{X}^{R,\epsilon}_t,\tilde{Y}^{R,\epsilon}_t,T-t)dW_t,\\
d\tilde{Y}^{R,\epsilon}_t &= (\epsilon^{-1}\tilde{g}+\epsilon^{-2}c)
(\tilde{X}^{R,\epsilon}_t,\tilde{Y}^{R,\epsilon}_t,T-t)dt
+\epsilon^{-1}\eta(\tilde{X}^{R,\epsilon}_t,\tilde{Y}^{R,\epsilon}_t,T-t)dW_t,
\end{split}\right.
\end{equation}
where $\tilde{b},\tilde{f}:\Rnum^m\times\Rnum^n\times\Rnum^+\rightarrow\Rnum^m$ and
$\tilde{g}:\Rnum^m\times\Rnum^n\times\Rnum^+\rightarrow\Rnum^n$. Unlike the forward case, the comparable process in the backward case is controlled by a reversed protocol since the temporal variable in the drift and diffusion coefficients is $T-t$ rather than $t$. The initial distribution of $(\tilde{X}^{R,\epsilon},\tilde{Y}^{R,\epsilon})$ (at time $0$) is chosen to be the same as the final distribution of $(X^\epsilon,Y^\epsilon)$ (at time $T$). We also assume that the comparable process satisfies the regularity conditions (a)-(d). The backward thermodynamic functional for the original process is defined rigorously as follows.

\begin{definition}\label{backwardfun1}
Let $P^\epsilon_T$ be the law of the original process $(X^\epsilon,Y^\epsilon)$, let $\tilde{P}^{R,\epsilon}_T$ be the law of the comparable process $(\tilde{X}^{R,\epsilon},\tilde{Y}^{R,\epsilon})$, and let $\tilde{Q}^{R,\epsilon}_T$ be the law of its time reversal $(\tilde{X}^{R,\epsilon}_{T-t},\tilde{Y}^{R,\epsilon}_{T-t})_{0\leq t\leq T}$. Then the backward thermodynamic functional of $(X^\epsilon,Y^\epsilon)$ with respect to $(\tilde{X}^{R,\epsilon},\tilde{Y}^{R,\epsilon})$ is defined as
\begin{equation*}
G^\epsilon_T = \log\frac{dP^\epsilon_T}{d\tilde{Q}^{R,\epsilon}_T}(X^\epsilon_\cdot,Y^\epsilon_\cdot),
\end{equation*}
if $P^\epsilon_T$ and $\tilde{Q}^{R,\epsilon}_T$ are absolutely continuous with respect to each other.
\end{definition}

Since the comparable process satisfies the regularity conditions (a)-(d), there exists a unique smooth solution $\tilde\phi$ to the Poisson equation
\begin{equation*}\left\{
\begin{split}
&-L_0\tilde\phi(x,y,t) = \tilde{f}(x,y,t),\\
&\int_{\Rnum^n}\tilde\phi(x,y,t)\rho(x,y,t)dy = 0.
\end{split}\right.
\end{equation*}
It then follows from Lemma \ref{convergence} that the slow component $\tilde{X}^{R,\epsilon}$ of the comparable process converges weakly in $C([0,T],\Rnum^m)$ to another diffusion $\tilde{X}^R = (\tilde{X}^R_t)_{0\leq t\leq T}$ with generator
\begin{equation*}
\tilde{L} = \tilde{w}^i(x,T-t)\partial_{x_i}+\frac{1}{2}\tilde{A}^{ij}(x,T-t)\partial_{x_ix_j},
\end{equation*}
where
\begin{gather*}
\tilde{w}^i(x,t) = \int_{\Rnum^n}(\tilde{b}^i+\partial_{x_j}\tilde\phi^i\tilde{f}^j
+\partial_{y_j}\tilde\phi^i\tilde{g}^j+\partial_{x_jy_k}\tilde\phi^ih^{jk})(x,y,t)\rho(x,y,t)dy,\\
\tilde{A}^{ij}(x,t) = \int_{\Rnum^n}(a^{ij}+\tilde\phi^i\tilde{f}^j+\tilde\phi^j\tilde{f}^i
+\partial_{y_k}\tilde\phi^ih^{jk}+\partial_{y_k}\tilde\phi^jh^{ik})(x,y,t)\rho(x,y,t)dy.
\end{gather*}
The backward thermodynamic functional for the averaged process is defined rigorously as follows.

\begin{definition}
Let $P_T$ be the law of $X$, let $\tilde{P}^R_T$ be the law of $\tilde{X}^R$, and let $\tilde{Q}^R_T$ be the law of its time reversal $(\tilde{X}^R_{T-t})_{0\leq t\leq T}$. Then the backward thermodynamic function of $X$ with respect to $\tilde{X}^R$ is defined as
\begin{equation*}\label{backwardfun2}
G^{(1)}_T = \log\frac{dP_T}{d\tilde{Q}^R_T}(X_\cdot),
\end{equation*}
if $P_T$ and $\tilde{Q}^R_T$ are absolutely continuous with respect to each other.
\end{definition}

Similarly, if we let the ending time $T$ change, then both $G^\epsilon_t$ and $G^{(1)}_t$ can be viewed as stochastic processes with $T$ being replaced by an arbitrary time $t$. In the backward case, we need the following two important assumptions.

\begin{assumption}\label{a1}
The function $\tilde{f}$ satisfies the compatible condition
\begin{equation*}
f+\tilde{f} = \nabla_y\cdot h+h\nabla_y\log\rho,
\end{equation*}
where $\nabla_y\cdot h$ is a vector-valued function whose $i$th component is given by $(\nabla_y\cdot h)^i = \partial_{y_j}h^{ij}$.
\end{assumption}

\begin{assumption}\label{a2}
The pseudo-stationary density $\rho$ satisfies the compatible condition
\begin{equation*}
2c = \nabla_y\cdot\alpha+\alpha\nabla_y\log\rho,
\end{equation*}
where $\nabla_y\cdot\alpha$ is a vector-valued function whose $i$th component is given by $(\nabla_y\cdot\alpha)^i = \partial_{y_j}\alpha^{ij}$.
\end{assumption}

Roughly speaking, Assumption \ref{a2} holds if and only if the diffusion $Y^x$ is reversible during each infinitesimal time period. It is a classical result that this assumption has the following equivalent statement \cite{jiang2004mathematical}.

\begin{lemma}\label{a2explain}
Assumption \ref{a2} holds if and only if the operator $L_0$ is symmetric with respect to the pseudo-stationary density $\rho$, i.e.
\begin{equation*}
(L_0f,g)_\rho = (f,L_0g)_\rho,\;\;\;f,g\in C_b^\infty(\Rnum^n),
\end{equation*}
where $C_b^\infty(\Rnum^n)$ is the space of all bounded smooth functions on $\Rnum^n$ whose all partial derivatives are bounded and
\begin{equation*}
(f,g)_\rho := \int_{\Rnum^n}f(y)g(y)\rho(x,y,t)dy.
\end{equation*}
\end{lemma}

\begin{remark}
From the proof in Section \ref{backwardproof}, especially the expression \eqref{notdiverge}, it is easy to see that the above two assumptions are the sufficient and necessary conditions that guarantee the backward thermodynamic function $G^\epsilon_T$ not to diverge as $\epsilon\rightarrow 0$. If $h = 0$ (which means that the noise terms of fast and slow components are not correlated), then Assumption \ref{a1} reduces to $\tilde{f} = -f$.

We emphasize that in this section, our theory is established when all slow and fast components of the original process are ``even variables" (this concept will be explained in Section \ref{odd}). However, in many physical systems, some variables of the original processes are actually ``odd variables" and thus the function $\tilde{f}$ should be replaced by another function $\tilde{f}_\delta$ (see Section \ref{odd} for details). For the classical model of underdamped Langevin equations, all slow components are even variables and all fast components are odd variables. In this case, we have $h = 0$ and $\tilde{f}_\delta = -f$, and thus Assumption \ref{a1} is automatically satisfied. This indicates that the above two assumptions are not excessive requirements for real physical systems.

Similarly, in the forward case, the reason why we use the same function $f$ in the comparable process instead of introducing a new function $\hat{f}$ is to guarantee the forward thermodynamic functional $F^\epsilon_T$ not to diverge as $\epsilon\rightarrow 0$.
\end{remark}

Interestingly, the above two assumptions also imply the following proposition.

\begin{proposition}\label{same}
Under Assumptions \ref{a1} and \ref{a2}, we have $A^{ij}(x,t) = \tilde{A}^{ij}(x,t)$, where $A^{ij}(x,t)$ is the diffusion coefficient of $X$ and $\tilde{A}^{ij}(x,T-t)$ is the diffusion coefficient of $\tilde{X}^R$.
\end{proposition}

\begin{proof}
Using Assumption \ref{a1} and integration by parts, we obtain
\begin{equation*}
\begin{split}
\int_{\Rnum^n}(\tilde\phi^i\tilde{f}^j+\partial_{y_k}\tilde\phi^ih^{jk})\rho dy
&= \int_{\Rnum^n}\tilde\phi^i[\tilde{f}^j\rho-\partial_{y_k}(h^{jk}\rho)]dy\\
&= \int_{\Rnum^n}\tilde\phi^i[\tilde{f}^j-\partial_{y_k}h^{jk}-h^{jk}\partial_{y_k}\log\rho]\rho dy
= -\int_{\Rnum^n}\tilde\phi^if^j\rho dy.
\end{split}
\end{equation*}
By Assumption \ref{a2} and Lemma \ref{a2explain}, $L^0$ is symmetric with respect to $\rho$ and thus
\begin{equation*}
\int_{\Rnum^n}\tilde\phi^if^j\rho dy = -(\tilde\phi^i,L_0\phi^j)_\rho
= -(L_0\tilde\phi^i,\phi^j)_\rho = \int_{\Rnum^n}\phi^j\tilde f^i\rho dy.
\end{equation*}
Using Assumption \ref{a1} and integration by parts again yields
\begin{equation*}
\begin{split}
\int_{\Rnum^n}(\tilde\phi^i\tilde{f}^j+\partial_{y_k}\tilde\phi^ih^{jk})\rho dy
&= \int_{\Rnum^n}\phi^j(f^i-\partial_{y_k}h^{ik}-h^{ik}\partial_{y_k}\log\rho)\rho dy\\
&= \int_{\Rnum^n}\phi^j[f^i\rho-\partial_{y_k}(h^{ik}\rho)]dy
= \int_{\Rnum^n}(\phi^jf^i+\partial_{y_k}\phi^jh^{ik})\rho dy.
\end{split}
\end{equation*}
Similarly, we can prove that
\begin{equation*}
\int_{\Rnum^n}(\tilde\phi^j\tilde{f}^i+\partial_{y_k}\tilde\phi^jh^{ik})\rho dy
= \int_{\Rnum^n}(\phi^if^j+\partial_{y_k}\phi^ih^{jk})\rho dy.
\end{equation*}
Combining the above two equations finally yields
\begin{equation*}
\begin{split}
\tilde{A}^{ij}(x,t) &= \int_{\Rnum^n}(a^{ij}+\tilde\phi^i\tilde{f}^j+\tilde\phi^j\tilde{f}^i
+\partial_{y_k}\tilde\phi^ih^{jk}+\partial_{y_k}\tilde\phi^jh^{ik})\rho dy\\
&= \int_{\Rnum^n}(a^{ij}+\phi^if^j+\phi^jf^i
+\partial_{y_k}\phi^ih^{jk}+\partial_{y_k}\phi^jh^{ik})\rho dy = A^{ij}(x,t),
\end{split}
\end{equation*}
which gives the desired result.
\end{proof}

For any $\epsilon>0$ and any Borel set $K\subset\Rnum^{n+m}$, let $\tau^\epsilon_K = \inf\{t\geq 0:(X^\epsilon_t,Y^\epsilon_t)\notin K\}$ be the first exit time of $(X^\epsilon,Y^\epsilon)$ from $K$. Since we are working with processes defined on the entire space, rather than on a compact subset as assumed in many previous papers \cite{khasminskii1996transition, khasminskii2004averaging, khasminskii2005limit}, we still need the following technical assumption.

\begin{assumption}\label{a3}
For any $\delta>0$, there exists a compact set $K\subset\Rnum^{n+m}$ such that
\begin{equation*}
\Pnum\left(\inf_{\epsilon>0}\tau^\epsilon_K\geq T\right) \geq 1-\delta.
\end{equation*}
\end{assumption}

The martingale structure for backward thermodynamic functionals under second-order averaging is stated as follows.

\begin{theorem}\label{backward}
Under Assumptions \ref{a1}, \ref{a2}, and \ref{a3}, the following conclusions hold:
\begin{itemize}\style
\item[(a)] For any $t\geq 0$, both $G^\epsilon_t$ and $G^{(1)}_t$ are well defined.
\item[(b)] The pair of processes $(X^\epsilon_t,G^\epsilon_t)_{0\leq t\leq T}$ converges weakly to another pair of processes $(X_t,G_t)_{0\leq t\leq T}$ in $C([0,T],\Rnum^{m+1})$.
\item[(c)] If we decompose the limit functional into $G_t = G^{(1)}_t+G^{(2)}_t$, then $e^{-G^{(2)}_t}$ is a martingale.
\end{itemize}
\end{theorem}

\begin{proof}
The proof of this theorem will be given in Section \ref{backwardproof}.
\end{proof}

Unlike the forward case, the backward thermodynamic functional $G^\epsilon_t$, the limit functional $G_t$, and its regular part $G^{(1)}_t$ are generally not martingales. In stochastic thermodynamics, most thermodynamic quantities can be represented as forward or backward thermodynamic functionals. In both the forward and backward cases, a general thermodynamic functional has a weak limit under second-order averaging, which can be decomposed into the regular part and the anomalous part. Theorems \ref{forward} and \ref{backward} reveal a highly nontrivial relationship between model simplification and martingale. In the forward case, the limit functional and its two parts are all martingales. In the backward case, while the limit functional and its regular part may not be martingales, the anomalous part is still a martingale. Therefore, our theory indicates that the anomalous part of a general thermodynamic functional is always a martingale under model simplification.

\begin{remark}
In this paper, our theorems are proved under a very general framework, which generalizes the setting in the physical literature in many ways. First, we reveal the martingale structure for a general thermodynamic functional defined as the logarithmic Radon-Nykodim derivative between the laws of the original process and a general comparable process (or its time reversal) whose drift terms $\hat{b}$ and $\hat{g}$ (or $\tilde{b}$ and $\tilde{g}$) can be arbitrarily chosen, rather than focusing on specific thermodynamic functionals such as entropy production and housekeeping heat. It can be seen from the proof that the arbitrariness of the drift terms in the comparable process greatly enhance the theoretical complexity. Second, we investigate model simplification under second-order averaging, which is physically more relevant and mathematically much more general and complicated than first-order averaging. In statistical physics, many important models such as underdamped Langevin equations are actually second-order singularly perturbed diffusions. Third, most previous papers focused on the case of diagonal diffusion matrices, while our theory can deal with the non-diagonal coupling case. In particular, our model covers the case of $h\neq 0$, which describes correlated noise terms between slow and fast components (such effect is particularly important in biochemical reaction networks).
\end{remark}

\begin{corollary}[\textbf{Integral fluctuation theorem}]
The backward functionals $G^\epsilon_t$, $G^{(1)}_t$, and $G^{(2)}_t$ satisfy the following integral fluctuation theorem:
\begin{itemize}\style
\item[(a)] For any $t\geq 0$, we have
\begin{equation*}
\Enum e^{-G^\epsilon_t} = \Enum e^{-G^{(1)}_t} = 1.
\end{equation*}
\item[(b)] For any bounded stopping time $\tau$ with respect to $\{\mathscr{F}_t\}$, we have
\begin{equation*}
\Enum e^{-G^{(2)}_\tau} = 1.
\end{equation*}
\end{itemize}
\end{corollary}

\begin{proof}
Since $e^{-G^{(2)}_t}$ is a martingale, part (b) follows from the optional sampling theorem. We next prove part (a). By Definition \ref{backwardfun1}, we have
\begin{equation*}
\Enum e^{-G^\epsilon_t}
= \Enum\frac{d\tilde{Q}^{R,\epsilon}_t}{dP^\epsilon_t}(X^\epsilon_\cdot,Y^\epsilon_\cdot)
= \int_{C([0,t],\Rnum^{m+n})}\frac{d\tilde{Q}^{R,\epsilon}_t}{dP^\epsilon_t}(w)dP^\epsilon_t(w) = 1.
\end{equation*}
Similarly, the equality $\Enum e^{-G^{(1)}_t} = 1$ can be proved by using Definition \ref{backwardfun2}.
\end{proof}

Unlike the forward case, the limit functional $G_t$ in the backward case may not satisfy the integral fluctuation theorem because the family of random variables $(e^{-G^\epsilon_t})_{\epsilon>0}$ may not be uniformly integrable.

\section{Physical applications in the absence of odd variables}
In stochastic thermodynamics, most thermodynamic quantities can be represented as forward or backward thermodynamic functionals. Therefore, our general theory can be applied to understand the martingale structure and prove the integral fluctuation theorem for such thermodynamic quantities under second-order averaging. Two of the most important thermodynamic quantities among them are entropy production and housekeeping heat, which characterize the irreversibility of a molecular system. To define them rigorously, we need the following concepts.

\subsection{Some concepts}
In stochastic thermodynamics, there are many important processes associated with a given diffusion process $Z$. Specifically, let $Z = (Z_t)_{0\leq t\leq T}$ be an $\Rnum^m$-valued inhomogeneous diffusion process with generator
\begin{equation*}
\mathcal{A}_t = b^i(x,t)\partial_i+\frac{1}{2}a^{ij}(x,t)\partial_{ij}.
\end{equation*}
For simplicity, we assume that the drift $b = (b^i)$ and diffusion matrix $a = (a^{ij})$ are bounded and smooth, all their partial derivatives are bounded, and $a$ is uniformly elliptic.

\begin{definition}
The diffusion process with generator $\mathcal{A}_{T-t}$ and with its initial distribution being chosen as the final distribution of $Z$ (at time $T$) is called the original process $Z$ under the reversed protocol.
\end{definition}

For each $t$, let $p(x,t)$ be the probability density of $Z_t$ and let $\mathcal{B}_t$ be the adjoint operator of $\mathcal{A}_t$ with respect to the probability density $p(x,t)$ satisfying
\begin{equation*}
(\mathcal{A}_tf,g)_p = (f,\mathcal{B}_tg)_p,\;\;\;f,g\in C_b^\infty(\Rnum^m).
\end{equation*}
It is easy to check that $\mathcal{B}_t$ is given by
\begin{equation*}
\mathcal{B}_t = (-b^i+\partial_ja^{ij}+a^{ij}\log p)\partial_i+\frac{1}{2}a^{ij}\partial_{ij}.
\end{equation*}
Clearly, $\mathcal{B}_t$ is the generator of an inhomogeneous diffusion process with drift $-b+\nabla\cdot a+a\nabla\log p$ and diffusion matrix $a$. The following lemma is a classical result \cite{jiang2004mathematical}.

\begin{lemma}
The time reversal $(Z_{T-t})_{0\leq t\leq T}$ of $Z$ is a diffusion process with generator $\mathcal{B}_{T-t}$.
\end{lemma}

For each $t$, the pseudo-stationary density $\mu(x,t)$ of $Z$ is defined as the solution to the equation
\begin{gather*}
\mathcal{A}_t^*\mu = 0,\;\;\;\int_{\Rnum^m}\mu(x,t)dx = 1,
\end{gather*}
where
\begin{equation*}
\mathcal{A}_t^* = (\partial_ja^{ij}-b^i)\partial_i+\frac{1}{2}a^{ij}\partial_{ij}
+\frac{1}{2}\partial_{ij}a^{ij}-\partial_ib^i
\end{equation*}
is the adjoint operator of $\mathcal{A}_t$ with respect to the Lebesgue measure. Suppose that the pseudo-stationary density exists and is unique for each $t$. Moreover, let $\mathcal{A}^\dag_t$ be the adjoint operator of $\mathcal{A}_t$ with respect to the pseudo-stationary density $\mu(x,t)$ satisfying
\begin{equation*}
(\mathcal{A}_tf,g)_\mu = (f,\mathcal{A}^\dag_tg)_\mu,\;\;\;f,g\in C_b^\infty(\Rnum^m).
\end{equation*}
It is easy to check that $\mathcal{A}^\dag_t$ is given by
\begin{equation*}
\mathcal{A}^\dag_t = (-b^i+\partial_ja^{ij}+a^{ij}\log\mu)\partial_i+\frac{1}{2}a^{ij}\partial_{ij}.
\end{equation*}
Clearly, $\mathcal{A}^\dag_t$ is the generator of an inhomogeneous diffusion process with drift $-b+\nabla\cdot a+a\nabla\log\mu$ and diffusion matrix $a$.

\begin{definition}\label{adjoint}
The diffusion process with generator $\mathcal{A}^\dag_t$ and with its initial distribution being chosen as the initial distribution of $Z$ is called the adjoint process of $Z$.
\end{definition}

It is easy to see that when $Z$ is homogeneous and stationary, its time-reversed process and adjoint process are exactly the same.

\subsection{Martingale structure for entropy production}
In statistical physics, the deviation of a molecular system from equilibrium is usually characterized by the concept of entropy production. From the trajectory perspective, the entropy production of the original and reduced processes is defined rigorously as follows \cite{jiang2004mathematical, seifert2005entropy}.

\begin{definition}
Let $P^\epsilon_T$ be the law of the original process $(X^\epsilon,Y^\epsilon)$ and let $Q^{R,\epsilon}_T$ be the law of the time reversal of $(X^\epsilon,Y^\epsilon)$ under the reversed protocol. Then the entropy production of the original process is defined as
\begin{equation*}
S^\epsilon_{tot}(T) = \log\frac{dP^\epsilon_T}{dQ^{R,\epsilon}_T}(X^\epsilon_\cdot,Y^\epsilon_\cdot),
\end{equation*}
if $P^\epsilon_T$ and $Q^{R,\epsilon}_T$ are absolutely continuous with respect to each other. The entropy production $S^{(1)}_{tot}(T)$ of the averaged process $X$ can be defined similarly.
\end{definition}

Under the reversed protocol, $(X^\epsilon,Y^\epsilon)$ becomes a diffusion process with drift $B(x,y,T-t)$ and diffusion matrix $D(x,y,T-t)$. This is a special case of the comparable process \eqref{comparable2} in the backward case with $\tilde{b} = b$ and $\tilde{g} = g$. Therefore, entropy production is a kind of backward thermodynamic functional. The next theorem follows directly from Theorem \ref{backward}.

\begin{theorem}\label{total}
Suppose that Assumption \ref{a3} and the following two compatible conditions are satisfied:
\begin{gather*}
2f = \nabla_y\cdot h+h\nabla_y\log\rho,\\
2c = \nabla_y\cdot\alpha+\alpha\nabla_y\log\rho.
\end{gather*}
Then the following conclusions hold:
\begin{itemize}\style
\item[(a)] For any $t\geq 0$, both $S^\epsilon_{tot}(t)$ and $S^{(1)}_{tot}(t)$ are well defined.
\item[(b)] The pair of processes $(X^\epsilon_t,S^\epsilon_{tot}(t))_{0\leq t\leq T}$ converges weakly in $C([0,T],\Rnum^{m+1})$ to another pair of processes $(X_t,S_{tot}(t))_{0\leq t\leq T}$.
\item[(c)] If we decompose the limit functional into $S_{tot}(t) = S^{(1)}_{tot}(t)+S^{(2)}_{tot}(t)$, then $e^{-S^{(2)}_{tot}(t)}$ is a martingale.
\item[(d)] For any $t\geq 0$ and bounded stopping time $\tau$ with respect to $\{\mathscr{F}_t\}$, we have the following integral fluctuation theorem:
    \begin{equation*}
    \Enum e^{-S^\epsilon_{tot}(t)} = \Enum e^{-S^{(1)}_{tot}(t)} = \Enum e^{-S^{(2)}_{tot}(\tau)} = 1.
    \end{equation*}
\end{itemize}
\end{theorem}

\begin{remark}
Recently, it has been reported that when the original process is homogeneous and stationary, its entropy production is a martingale \cite{neri2017statistics}. Here we prove that under model simplification, the anomalous part of entropy production is always a martingale and hence satisfies the integral fluctuation theorem, whether the original process is homogeneous and stationary or not.
\end{remark}

\subsection{Martingale structure for housekeeping heat}
From the trajectory perspective, the housekeeping heat of the original and reduced processes is defined rigorously as follows.

\begin{definition}
Let $P^\epsilon_T$ be the law of the original process $(X^\epsilon,Y^\epsilon)$ and let $\hat{P}^\epsilon_{T}$ be the law of its adjoint process. Then the housekeeping heat, also called adiabatic entropy production, of the original process is defined as
\begin{equation*}
S^\epsilon_{hk}(T) = \log\frac{dP^\epsilon_T}{d\hat{P}^\epsilon_T}(X^\epsilon_\cdot,Y^\epsilon_\cdot),
\end{equation*}
if $P^\epsilon_T$ and $\hat{P}^\epsilon_T$ are absolutely continuous with respect to each other. The housekeeping heat $S^{(1)}_{hk}(T)$ of the averaged process $X$ can be defined similarly.
\end{definition}

For each $t$, let $\mu^\epsilon(x,y,t)$ be the the pseudo-stationary density of $(X^\epsilon,Y^\epsilon)$. Suppose that the pseudo-stationary density exists and is unique. As discussed earlier, the adjoint process of $(X^\epsilon,Y^\epsilon)$ is a diffusion process with drift
\begin{equation*}
\hat{B}(x,y,t) =
\begin{pmatrix}
\hat{b}(x,y,t)+\epsilon^{-1}\hat{f}(x,y,t) \\
\epsilon^{-1}\hat{g}(x,y,t)+\epsilon^{-2}\hat{c}(x,y,t)
\end{pmatrix}
\end{equation*}
and diffusion matrix $D(x,y,t)$, where
\begin{gather*}
\hat{b} = -b+\nabla_x\cdot a+a\nabla_x\log\mu^\epsilon,\\
\hat{f} = -f+\nabla_y\cdot h+h\nabla_y\log\mu^\epsilon,\\
\hat{g} = -g+\nabla_x\cdot h'+h'\nabla_x\log\mu^\epsilon,\\
\hat{c} = -c+\nabla_y\cdot\alpha+\alpha\nabla_y\log\mu^\epsilon.
\end{gather*}
Under mild conditions, by using the matched asymptotic expansions of singularly perturbed diffusion processes \cite{khasminskii1996transition, khasminskii2004averaging, khasminskii2005limit}, it can be proved that
\begin{gather*}
\mu^\epsilon(x,y,t) = \mu(x,t)\rho(x,y,t)+O(\epsilon),\\
\nabla_x\log\mu^\epsilon(x,y,t) = \nabla_x\log\mu(x,t)+\nabla_x\log\rho(x,y,t)+O(\epsilon),\\
\nabla_y\log\mu^\epsilon(x,y,t) = \nabla_y\log\rho(x,y,t)+O(\epsilon),
\end{gather*}
where $\mu(x,t)$ is the pseudo-stationary density of the averaged process $X$ which is also assumed to exist and be unique for each $t$. Therefore, the drift coefficients have the following asymptotic expansions:
\begin{gather*}
\hat{b} = -b+\nabla_x\cdot a+a\nabla_x\log\mu+a\nabla_x\log\rho+O(1),\\
\hat{f} = -f+\nabla_y\cdot h+h\nabla_y\log\rho,\\
\hat{g} = -g+\nabla_x\cdot h'+h'\nabla_x\log\mu+h'\nabla_x\log\rho+O(1),\\
\hat{c} = -c+\nabla_y\cdot\alpha+\alpha\nabla_y\log\rho,
\end{gather*}
where we have moved the $O(\epsilon)$ terms of $\hat{f}$ and $\hat{c}$ into $\hat{b}$ and $\hat{g}$, respectively. If Assumption \ref{a2} is satisfied and and if we ignore the $O(1)$ terms of $\hat{b}$ and $\hat{g}$ (see Remark \ref{infinitesimal} for its justification), then we have $\hat{c} = c$ and thus the adjoint process of $(X^\epsilon,Y^\epsilon)$ is a special case of the comparable process \eqref{comparable1} in the forward case. Therefore, housekeeping heat is a kind of forward thermodynamic functional. The next theorem follows directly from Theorem \ref{forward}.

\begin{theorem}\label{adiabatic}
Suppose that the following two compatible conditions are satisfied:
\begin{gather*}
2f = \nabla_y\cdot h+h\nabla_y\log\rho,\\
2c = \nabla_y\cdot\alpha+\alpha\nabla_y\log\rho.
\end{gather*}
Then the following conclusions hold:
\begin{itemize}\style
\item[(a)] For any $t\geq 0$, both $S^\epsilon_{hk}(t)$ and $S^{(1)}_{hk}(t)$ are well defined.
\item[(b)] The pair of processes $(X^\epsilon_t,S^\epsilon_{hk}(t))_{0\leq t\leq T}$ converges weakly in $C([0,T],\Rnum^{m+1})$ to another pair of processes $(X_t,S_{hk}(t))_{0\leq t\leq T}$.
\item[(c)] Both $e^{-S^\epsilon_{hk}(t)}$ and $e^{-S_{hk}(t)}$ are martingales.
\item[(d)] If we decompose the limit functional into $S_{hk}(t) = S^{(1)}_{hk}(t)+S^{(2)}_{hk}(t)$, then $e^{-S^{(1)}_{hk}(t)}$ and $e^{-S^{(2)}_{hk}(t)}$ are orthogonal martingales.
\item[(e)] For any bounded stopping time $\tau$ with respect to $\{\mathscr{F}_t\}$, we have the following integral fluctuation theorem:
    \begin{equation*}
    \Enum e^{-S^\epsilon_{hk}(\tau)} = \Enum e^{-S_{hk}(\tau)} = \Enum e^{-S^{(1)}_{hk}(\tau)} = \Enum e^{-S^{(2)}_{hk}(\tau)} = 1.
    \end{equation*}
\end{itemize}
\end{theorem}

\begin{remark}
Recently, it has been reported that the housekeeping heat of the original process is a martingale \cite{chetrite2019martingale}. Here we prove that under model simplification, the regular and anomalous parts of housekeeping heat are always orthogonal martingales and hence satisfy the integral fluctuation theorem.
\end{remark}

\section{Physical applications in the presence of odd variables}\label{odd}

\subsection{General theorem}
In statistical physics, physical quantities are often divided into even and odd variables. For example, position, kinetic energy, and potential energy are even variables, while velocity and Lorentz force are odd variables. Intuitively, if a smooth function $z(t)$ describes the trajectory of a classical particle, then its time reversal should be $z(T-t)$. The velocity of the original trajectory is $\dot{z}(t)$ and the velocity of the time-reversed trajectory is
\begin{equation*}
\frac{d}{dt}z(T-t) = -\dot{z}(T-t).
\end{equation*}
This indicates that the time reversal of the velocity process $\dot{z}(t)$ should be defined as $-\dot{z}(T-t)$ rather than $\dot{z}(T-t)$. Such kind of variables are called odd variables. Specifically, we give the following definition.

\begin{definition}\label{evenodd}
Let $Z = (Z_t)_{0\leq t\leq T}$ be a one-dimensional stochastic process.
\begin{itemize}\style
\item[(a)] If the time reversal of $Z$ is defined as $(Z_{T-t})_{0\leq t\leq T}$, then $Z$ is called an even variable.
\item[(b)] If the time reversal of $Z$ is defined as $(-Z_{T-t})_{0\leq t\leq T}$, then $Z$ is called an odd variable.
\end{itemize}
\end{definition}

Since the definition of forward thermodynamic functionals is irrelevant of time reversal, the results in the forward case will not change when odd variables are considered. Unlike the forward case, the definition of backward thermodynamic functionals depends on the time reversal of the comparable process. Therefore, all the results in the backward case discussed earlier only hold for systems whose components are all even variables and thus should be modified in the presence of odd variables.

From now on, we focus on the general case when each component of the original process $(X^\epsilon,Y^\epsilon)$ can be either an even or an odd variable. In the sense of Definition \ref{evenodd}, the time reversal of the original process will be denoted by $(\delta X^\epsilon_{T-t},\delta Y^\epsilon_{T-t})_{0\leq t\leq T}$, where
\begin{equation*}
\delta = (\delta_{x_1},\cdots,\delta_{x_m},\delta_{y_1},\cdots,\delta_{y_n})
\end{equation*}
with each component of $\delta$ taking the value of $1$ or $-1$, corresponding to an even or an odd variable respectively.

In the backward case, let $(\tilde{X}^{R,\epsilon},\tilde{Y}^{R,\epsilon})$ be the comparable process given in \eqref{comparable2}. Let $P^\epsilon_T$ be the law of the original process $(X^\epsilon,Y^\epsilon)$ and let $\tilde{Q}^{R,\epsilon}_T$ be the law of the time reversal of the comparable process in the sense of Definition \ref{evenodd}. In other words, $\tilde{Q}^{R,\epsilon}_T$ is the law of the time reversal of $(\delta\tilde{X}^{R,\epsilon},\delta\tilde{Y}^{R,\epsilon})$ in the usual sense. The backward thermodynamic functional of $(X^\epsilon,Y^\epsilon)$ with respect to $(\tilde{X}^{R,\epsilon},\tilde{Y}^{R,\epsilon})$ is still defined as
\begin{equation*}
G^\epsilon_T = \log\frac{dP^\epsilon_T}{d\tilde{Q}^{R,\epsilon}_T}(X^\epsilon_\cdot,Y^\epsilon_\cdot).
\end{equation*}
It then follows from Lemma \ref{convergence} that $X^\epsilon\Rightarrow X$ and $\tilde{X}^{R,\epsilon}\Rightarrow\tilde{X}^R$ in $C([0,T],\Rnum^m)$ as $\epsilon\rightarrow 0$. Similarly, let $P_T$ be the law of the averaged process $X$ and let $\tilde{Q}^R_T$ be the law of the time reversal of $\tilde{X}^R$ in the sense of Definition \ref{evenodd}. The backward thermodynamic function of $X$ with respect to $\tilde{X}^R$ is still defined as
\begin{equation*}
G^{(1)}_T = \log\frac{dP_T}{d\tilde{Q}^R_T}(X_\cdot).
\end{equation*}
Clearly, $(\delta\tilde{X}^{R,\epsilon},\delta\tilde{Y}^{R,\epsilon})$ is a diffusion process with drift
\begin{equation*}
\tilde{B}_\delta =
\begin{pmatrix}
\tilde{b}_\delta(x,y,T-t)+\epsilon^{-1}\tilde{f}_\delta(x,y,T-t) \\
\epsilon^{-1}\tilde{g}_\delta(x,y,T-t)+\epsilon^{-2}c_\delta(x,y,T-t)
\end{pmatrix}.
\end{equation*}
and diffusion matrix
\begin{equation*}
D_\delta =
\begin{pmatrix}
a_\delta(x,y,T-t) & \epsilon^{-1}h_\delta(x,y,T-t) \\
\epsilon^{-1}h_\delta'(x,y,T-t) & \epsilon^{-2}\alpha_\delta(x,y,T-t)
\end{pmatrix},
\end{equation*}
where
\begin{gather*}
\tilde{b}^i_\delta(x,y,t) = \delta_{x_i}\tilde{b}^i(\delta x,\delta y,t),\;\;\;
\tilde{f}^i_\delta(x,y,t) = \delta_{x_i}\tilde{f}^i(\delta x,\delta y,t),\\
\tilde{g}^i_\delta(x,y,t) = \delta_{y_i}\tilde{g}^i(\delta x,\delta y,t),\;\;\;
c^i_\delta(x,y,t) = \delta_{y_i}c^i(\delta x,\delta y,t),\\
a^{ij}_\delta(x,y,t) = \delta_{x_i}\delta_{x_j}a^{ij}(\delta x,\delta y,t),\\
h^{ij}_\delta(x,y,t) = \delta_{x_i}\delta_{y_j}h^{ij}(\delta x,\delta y,t),\\
\alpha^{ij}_\delta(x,y,t) = \delta_{y_i}\delta_{y_j}\alpha^{ij}(\delta x,\delta y,t).
\end{gather*}
To make $(\delta\tilde{X}^{R,\epsilon},\delta\tilde{Y}^{R,\epsilon})$ have the form of the comparable process \eqref{comparable2}, we need to impose the following compatible conditions.

\begin{assumption}\label{a4}
\begin{equation*}
c_\delta = c,\;\;\;a_\delta = a,\;\;\;h_\delta = h,\;\;\;\alpha_\delta = \alpha.
\end{equation*}
\end{assumption}

In fact, this assumption guarantees that the laws $P^\epsilon_T$ and $\tilde{Q}^{R,\epsilon}_T$ are absolutely continuous with respect to each other. In most cases of physical interests, we have $h = 0$ (which means that the noise terms of fast and slow components are not correlated), the slow components are all even variables, and the fast components are all odd variables \cite{spinney2012entropy, lee2013fluctuation, ge2014time}. In this case, we automatically have
\begin{equation}\label{trivial}
a_\delta = a,\;\;\;h_\delta = h,\;\;\;\alpha_\delta = \alpha.
\end{equation}

\begin{theorem}\label{backwardodd}
In the presence of odd variables, suppose that Assumptions \ref{a3} and \ref{a4} and the following two compatible conditions are satisfied:
\begin{gather*}
f+\tilde{f}_\delta = \nabla_y\cdot h+h\nabla_y\log\rho,\\
2c = \nabla_y\cdot\alpha+\alpha\nabla_y\log\rho.
\end{gather*}
Then the conclusions in Theorem \ref{backward} hold.
\end{theorem}

\subsection{Martingale structure for entropy production}
Let $P^\epsilon_T$ be the law of the original process $(X^\epsilon,Y^\epsilon)$ and let $Q^{R,\epsilon}_T$ be the law of the time reversal in the sense of Definition \ref{evenodd} of $(X^\epsilon,Y^\epsilon)$ under the reversed protocol. In other words, $Q^{R,\epsilon}_T$ is the law of the time reversal in the usual sense of $(\delta X^\epsilon,\delta Y^\epsilon)$ under the reversed protocol. In the presence of odd variables, the entropy production of the original process is still defined as
\begin{equation*}
S^\epsilon_{tot}(T) = \log\frac{dP^\epsilon_T}{dQ^{R,\epsilon}_T}(X^\epsilon_\cdot,Y^\epsilon_\cdot).
\end{equation*}
The entropy production $S^{(1)}_{tot}(T)$ of the averaged process $X$ can be defined similarly. Under the reversed protocol, $(\delta X^\epsilon,\delta Y^\epsilon)$ becomes a diffusion process with drift
\begin{equation*}
B_\delta =
\begin{pmatrix}
b_\delta(x,y,T-t)+\epsilon^{-1}f_\delta(x,y,T-t) \\
\epsilon^{-1}g_\delta(x,y,T-t)+\epsilon^{-2}c_\delta(x,y,T-t)
\end{pmatrix}
\end{equation*}
and diffusion matrix $D_\delta$, where
\begin{gather*}
b^i_\delta(x,y,t) = \delta_{x_i}b^i(\delta x,\delta y,t),\;\;\;
f^i_\delta(x,y,t) = \delta_{x_i}f^i(\delta x,\delta y,t),\\
g^i_\delta(x,y,t) = \delta_{y_i}g^i(\delta x,\delta y,t),\;\;\;
c^i_\delta(x,y,t) = \delta_{y_i}c^i(\delta x,\delta y,t).
\end{gather*}
If Assumption \ref{a4} is satisfied, then we have $c_\delta = c$ and $D_\delta = D$, and thus is a special case of the comparable process \eqref{comparable2} in the backward case with $\tilde{b} = b_\delta$ and $\tilde{g} = g_\delta$. The next theorem follows directly from Theorem \ref{backwardodd}.

\begin{theorem}
In the presence of odd variables, suppose that Assumptions \ref{a3} and \ref{a4} and the following two compatible conditions are satisfied:
\begin{gather*}
f+f_\delta = \nabla_y\cdot h+h\nabla_y\log\rho,\\
2c = \nabla_y\cdot\alpha+\alpha\nabla_y\log\rho.
\end{gather*}
Then the conclusions in Theorem \ref{total} hold.
\end{theorem}

\begin{example}\label{underdamp}
In statistical physics, a classical model is the following underdamped Langevin equation under the zero-mass limit \cite{spinney2012entropy}:
\begin{equation*}\left\{
\begin{split}
dX^\epsilon_t &= \epsilon^{-1}Y^{\varepsilon}_t,\\
dY^\epsilon_t &= [\epsilon^{-1}g(X^\epsilon_t,V^\epsilon_t,t)-\epsilon^{-2}\gamma Y^\epsilon_t]dt
+\epsilon^{-1}\eta dW_t,
\end{split}\right.
\end{equation*}
where $X^\epsilon$ represents the position of a mesoscopic particle, $Y^\epsilon$ represents its velocity, $\gamma$ is the fraction coefficient of the particle with the fluid, and $\eta$ is a constant matrix. In this case, the slow components are all even variables and the fast components are all odd variables. Clearly, the relations \eqref{trivial} hold for this model with $h = 0$. Since $f(x,y,t) = y$ and $c(x,y,t) = -\gamma y$, we also have
\begin{gather*}
f_\delta(x,y,t) = -y = -f(x,y,t),\\
c_\delta(x,y,t) = -\gamma y = c(x,y,t).
\end{gather*}
Therefore, for the underdamped Langevin equation, both Assumption \ref{a4} and the first compatible condition in the above theorem are automatically satisfied.
\end{example}

\subsection{Martingale structure for housekeeping heat}
Since the adjoint process coincides with the time reversal when the original process is homogenous and stationary, we should also modify the definition of the adjoint process in the presence of odd variables.

\begin{definition}\label{adjointodd}
In the presence of odd variables, the adjoint process of $(X^\epsilon,Y^\epsilon)$ is defined as the adjoint process of $(\delta X^\epsilon,\delta Y^\epsilon)$ in the sense of Definition \ref{adjoint}.
\end{definition}

Let $P^\epsilon_T$ be the law of the original process $(X^\epsilon,Y^\epsilon)$ and let $\hat{P}^\epsilon_T$ be the law of its adjoint process in the sense of Definition \ref{adjointodd}. In other words, $\hat{P}^\epsilon_T$ is the law of the adjoint process of $(\delta X^\epsilon,\delta Y^\epsilon)$ in the usual sense. In the presence of odd variables, the housekeeping heat of the original process is still defined as
\begin{equation*}
S^\epsilon_{hk}(T) = \log\frac{dP^\epsilon_T}{d\hat{P}^\epsilon_T}(X^\epsilon_\cdot,Y^\epsilon_\cdot).
\end{equation*}
The housekeeping heat $S^{(1)}_{hk}(T)$ of the averaged process $X$ can be defined similarly.

For each $t$, let $\mu^\epsilon(x,y,t)$ be the pseudo-stationary density of $(\delta X^\epsilon,\delta Y^\epsilon)$. As discussed earlier, the adjoint process of $(\delta X^\epsilon,\delta Y^\epsilon)$ in the usual sense is a diffusion process with drift
\begin{equation*}
\hat{B}_\delta =
\begin{pmatrix}
\hat{b}_\delta+\epsilon^{-1}\hat{f}_\delta \\
\epsilon^{-1}\hat{g}_\delta+\epsilon^{-2}\hat{c}_\delta
\end{pmatrix}
\end{equation*}
and diffusion matrix $D_\delta$, where
\begin{gather*}
\hat{b}_\delta = -b_\delta+\nabla_x\cdot a+a\nabla_x\log\mu^\epsilon,\\
\hat{f}_\delta = -f_\delta+\nabla_y\cdot h+h\nabla_y\log\mu^\epsilon,\\
\hat{g}_\delta = -g_\delta+\nabla_x\cdot h'+h'\nabla_x\log\mu^\epsilon,\\
\hat{c}_\delta = -c_\delta+\nabla_y\cdot\alpha+\alpha\nabla_y\log\mu^\epsilon.
\end{gather*}
Under mild conditions, by using the matched asymptotic expansions of singularly perturbed diffusion processes \cite{khasminskii1996transition, khasminskii2004averaging, khasminskii2005limit}, it can be proved that
\begin{gather*}
\mu^\epsilon(x,y,t) = \mu(x,t)\rho(x,y,t)+O(\epsilon),\\
\nabla_x\log\mu^\epsilon(x,y,t) = \nabla_x\log\mu(x,t)+\nabla_x\log\rho(x,y,t)+O(\epsilon),\\
\nabla_y\log\mu^\epsilon(x,y,t) = \nabla_y\log\rho(x,y,t)+O(\epsilon),
\end{gather*}
where $\mu(x,t)$ is the pseudo-stationary density of the averaged process $\delta X$. Therefore, the drift coefficients have the following asymptotic expansions:
\begin{gather*}
\hat{b}_\delta = -b_\delta+\nabla_x\cdot a+a\nabla_x\log\mu+a\nabla_x\log\rho+O(1),\\
\hat{f}_\delta = -f_\delta+\nabla_y\cdot h+h\nabla_y\log\rho,\\
\hat{g}_\delta = -g_\delta+\nabla_x\cdot h'+h'\nabla_x\log\mu+h'\nabla_x\log\rho+O(1),\\
\hat{c}_\delta = -c_\delta+\nabla_y\cdot\alpha+\alpha\nabla_y\log\rho,
\end{gather*}
where we have moved the $O(\epsilon)$ terms of $\hat{f}_\delta$ and $\hat{c}_\delta$ into $\hat{b}_\delta$ and $\hat{g}_\delta$, respectively. If Assumptions \ref{a2} and \ref{a4} are satisfied, then we have $\hat{c}_\delta = c$ and $D_\delta = D$, and thus the adjoint process of $(\delta X^\epsilon,\delta Y^\epsilon)$ in the usual sense is a special case of the comparable process \eqref{comparable1} in the forward case. The next theorem follows directly from Theorem \ref{forward}.

\begin{theorem}
In the presence of odd variables, suppose that Assumption \ref{a4} and the following two compatible conditions are satisfied:
\begin{gather*}
f+f_\delta = \nabla_y\cdot h+h\nabla_y\log\rho,\\
2c = \nabla_y\cdot\alpha+\alpha\nabla_y\log\rho.
\end{gather*}
Then the conclusions in Theorem \ref{adiabatic} hold.
\end{theorem}

\section{Proof of Theorem \ref{forward}}\label{forwardproof}
Recall that the original process $(X^\epsilon,Y^\epsilon)$ is a diffusion process with drift $B(x,y,t)$ and diffusion matrix $D(x,y,t)$. Moreover, the comparable process $(\hat{X}^\epsilon,\hat{Y}^\epsilon)$ is another diffusion process with drift
\begin{equation*}
\hat{B}(x,y,t) = \begin{pmatrix}
\hat{b}(x,y,t)+\epsilon^{-1}f(x,y,t)\\
\epsilon^{-1}\hat{g}(x,y,t)+\epsilon^{-2}c(x,y,t) \end{pmatrix}
\end{equation*}
and the same diffusion matrix $D(x,y,t)$. Without loss of generality, we assume that $D$ is invertible.

\begin{proof}[Proof of Theorem \ref{forward}]
Since $(X^\epsilon,Y^\epsilon)$ and $(\hat{X}^\epsilon,\hat{Y}^\epsilon)$ have the same diffusion matrix $D(x,y,t)$, it follows from Girsanov's theorem that their laws are absolutely continuous with respect to each other and
\begin{equation*}
\frac{dP^\epsilon_t}{d\hat{P}^\epsilon_t}(X^\epsilon_\cdot,Y^\epsilon_\cdot)
= e^{\int_0^t(B-\hat{B})'D^{-1}\Sigma(X^\epsilon_s,Y^\epsilon_s,s)dW_s
+\frac{1}{2}\int_0^t(B-\hat{B})'D^{-1}(B-\hat{B})(X^\epsilon_s,Y^\epsilon_s,s)ds},
\end{equation*}
where
\begin{equation*}
\Sigma(x,y,t) = \begin{pmatrix}
\sigma(x,y,t) \\ \epsilon^{-1}\eta(x,y,t)
\end{pmatrix}.
\end{equation*}
Therefore, the thermodynamic functional $F^\epsilon_t$ has the following explicit expression:
\begin{equation*}
F^\epsilon_t = \int_0^t(B-\hat{B})'D^{-1}\Sigma(X^\epsilon_s,Y^\epsilon_s,s)dW_s
+\frac{1}{2}\int_0^t(B-\hat{B})'D^{-1}(B-\hat{B})(X^\epsilon_s,Y^\epsilon_s,s)ds,
\end{equation*}
and it is clear that $e^{-F^\epsilon_t}$ is an exponential martingale. Moreover, it is easy to see that both $(B-\hat{B})'D^{-1}\Sigma$ and $(B-\hat{B})'D^{-1}(B-\hat{B})$ are independent of $\epsilon$. As a result, the ordered triple $(X^\epsilon,F^\epsilon,Y^\epsilon)$ is a diffusion process solving the SDE
\begin{equation*}\left\{
\begin{split}
dX^{\epsilon}_t &= [b(X^\epsilon_t,Y^\epsilon_t,t)+\epsilon^{-1}f(X^\epsilon_t,Y^\epsilon_t,t)]dt
+\sigma(X^\epsilon_t,Y^\epsilon_t,t)dW_t,\\
dF^{\epsilon}_t &= \frac{1}{2}(B-\hat{B})'D^{-1}(B-\hat{B})(X^\epsilon_t,Y^\epsilon_t,t)dt
+(B-\hat{B})'D^{-1}\Sigma(X^\epsilon_t,Y^\epsilon_t,t)dW_t,\\
dY^{\epsilon}_t
&= [\epsilon^{-1}g(X^\epsilon_t,Y^\epsilon_t,t)+\epsilon^{-2}c(X^\epsilon_t,Y^\epsilon_t,t)]dt
+\epsilon^{-1}\eta(X^\epsilon_t,Y^\epsilon_t,t)dW_t,
\end{split}\right.
\end{equation*}
where $X^\epsilon$ and $F^\epsilon$ are slowly varying and $Y^\epsilon$ is rapidly varying. This SDE can be rewritten as
\begin{equation*}\left\{
\begin{split}
&d\begin{pmatrix}X^{\epsilon}_t \\ F^{\epsilon}_t\end{pmatrix}
= [\bar{b}(X^\epsilon_t,Y^\epsilon_t,t)+\epsilon^{-1}\bar{f}(X^\epsilon_t,Y^\epsilon_t,t)]dt
+\bar\sigma(X^{\epsilon}_t,Y^{\epsilon}_t,t)dW_t,\\
&dY^{\epsilon}_t = [\epsilon^{-1}g(X^\epsilon_t,Y^\epsilon_t,t)+\epsilon^{-2}c(X^\epsilon_t,Y^\epsilon_t,t)]dt
+\epsilon^{-1}\eta(X^\epsilon_t,Y^\epsilon_t,t)dW_t,
\end{split}\right.
\end{equation*}
where
\begin{equation*}
\bar{b} = \begin{pmatrix} b \\ \frac{1}{2}(B-\hat{B})'D^{-1}(B-\hat{B}) \end{pmatrix},\;\;\;
\bar{f} = \begin{pmatrix} f \\ 0 \end{pmatrix},\;\;\;
\bar{\sigma} = \begin{pmatrix} \sigma \\ (B-\hat{B})'D^{-1}\Sigma \end{pmatrix}.
\end{equation*}
The diffusion matrix of $(X^\epsilon,F^\epsilon,Y^\epsilon)$ is given by
\begin{equation*}
\begin{pmatrix} \bar{\sigma} \\ \epsilon^{-1}\eta \end{pmatrix}
\begin{pmatrix} \bar{\sigma} \\ \epsilon^{-1}\eta \end{pmatrix}'
= \begin{pmatrix} \bar{a} & \epsilon^{-1}\bar{h} \\
\epsilon^{-1}\bar{h}' & \epsilon^{-2}\alpha \end{pmatrix},
\end{equation*}
where
\begin{equation*}
\bar{a} = \begin{pmatrix} a & b-\hat{b} \\
(b-\hat{b})' & (B-\hat{B})'D^{-1}(B-\hat{B}) \end{pmatrix},\;\;\;
\bar{h} = \begin{pmatrix} h \\ (g-\hat{g})' \end{pmatrix}.
\end{equation*}
Since $\bar{f}^{m+1} = 0$, the solution to the Poisson equation
\begin{equation*}\left\{
\begin{split}
&-L_0\bar\phi(x,y,t) = \bar{f}(x,y,t),\\
&\int_{\Rnum^n}\bar\phi(x,y,t)\rho(x,y,t)dy = 0
\end{split}\right.
\end{equation*}
is given by $\bar\phi = (\phi,0)'$. It thus follows from Lemma \ref{convergence} that $(X^\epsilon,F^\epsilon)\Rightarrow(X,F)$ in $C([0,T],\Rnum^{m+1})$, where $(X,F)$ is a diffusion process with generator
\begin{equation*}
\bar{L} = \sum_{i=1}^{m+1}\bar{w}^i(x,t)\partial_{x_i}+\frac{1}{2}\sum_{i,j=1}^{m+1}\bar{A}^{ij}(x,t)\partial_{x_ix_j},
\end{equation*}
where
\begin{gather*}
\bar{w}^i(x,t) = \int_{\Rnum^n}(\bar{b}^i+\partial_{x_j}\bar\phi^if^j+\partial_{y_j}\bar\phi^ig^j
+\partial_{x_jy_k}\bar\phi^ih^{jk})(x,y,t)\rho(x,y,t)dy,\\
\bar{A}^{ij}(x,t) = \int_{\Rnum^n}(\bar{a}^{ij}+\bar\phi^i\bar{f}^j+\bar\phi^j\bar{f}^i
+\partial_{y_k}\bar\phi^i\bar{h}^{jk}+\partial_{y_k}\bar\phi^j\bar{h}^{ik})(x,y,t)\rho(x,y,t)dy.
\end{gather*}
Since $\bar{f}^{m+1} = \bar\phi^{m+1} = 0$, direct computations show that
\begin{gather*}
\bar{w}^i(x,t) = w^i(x,t),\;\;\;\bar{A}^{ij}(x,t) = A^{ij}(x,t),\;\;\;1\leq i,j\leq m,\\
\bar{w}^{m+1}(x,t) = \frac{1}{2}\int_{\Rnum^n}(B-\hat{B})'D^{-1}(B-\hat{B})(x,y,t)\rho(x,y,t)dy,\\
\bar{A}^{m+1,j}(x,t) = \int_{\Rnum^n}(b^j-\hat{b}^j+\partial_{y_k}\phi^jg^k
-\partial_{y_k}\phi^j\hat{g}^k)(x,y,t)\rho(x,y,t)dy,\;\;\;1\leq j\leq m,\\
\bar{A}^{m+1,m+1}(x,t) = \int_{\Rnum^n}(B-\hat{B})'D^{-1}(B-\hat{B})(x,y,t)\rho(x,y,t)dy.
\end{gather*}
Therefore, $(X,F)$ can be viewed as the solution to the SDE
\begin{equation*}
\begin{cases}
dX^i_t = w^i(X_t,t)dt+(\bar{A}^{1/2})^{x,\cdot}(X_t,t)dB_t,\\
dF_t = \bar{w}^{m+1}(X_t,t)dt+(\bar{A}^{1/2})^{m+1,\cdot}(X_t,t)dB_t,
\end{cases}
\end{equation*}
where $\bar{A} = (\bar{A}^{ij})$ is an $(m+1)\times(m+1)$ matrix, $(\bar{A}^{1/2})^{x,\cdot}$ is the first $m$ rows of $\bar{A}^{1/2}$, $(\bar{A}^{1/2})^{m+1,\cdot}$ is the last row of $\bar{A}^{1/2}$, and $B = (B_t)_{t\geq 0}$ is an $(m+1)$-dimensional standard Brownian motion defined on some probability space. Since the diffusion matrix of $X$ is $A = (A^{ij})$, we have $\bar{A}^{x,x} = A$, where $\bar{A}^{x,x}$ is the matrix obtained from $\bar{A}$ by retaining the first $m$ rows and first $m$ columns. Let $\hat{X}$ be the averaged process of the comparable process $(\hat{X}^\epsilon,\hat{Y}^\epsilon)$. Since $X$ and $\hat{X}$ have the same diffusion matrix $A(x,t)$, it follows from Girsanov's theorem that their laws are absolutely continuous with respect to each other and
\begin{equation*}
\frac{dP_t}{d\hat{P}_t}(X_\cdot)
= e^{\int_0^t(w-\hat{w})'A^{-1}(\bar{A}^{1/2})^{x,\cdot}(X_s,s)dB_s
+\frac{1}{2}\int_0^t(w-\hat{w})'A^{-1}(w-\hat{w})(X_s,s)ds}.
\end{equation*}
Therefore, the thermodynamic functional $F^{(1)}_t$ has the following explicit expression:
\begin{equation*}
F^{(1)}_t = \int_0^t(w-\hat{w})'A^{-1}(\bar{A}^{1/2})^{x,\cdot}(X_s,s)dB_s
+\frac{1}{2}\int_0^t(w-\hat{w})'A^{-1}(w-\hat{w})(X_s,s)ds,
\end{equation*}
and it is clear that $e^{-F^{(1)}_t}$ is an exponential martingale. Obviously, we have
\begin{equation*}
F_t = \int_0^t(\bar{A}^{1/2})^{m+1,\cdot}(X_s,s)dB_s+\int_0^t\bar{w}^{m+1}(X_s,s)ds.
\end{equation*}
Therefore, we obtain
\begin{equation*}
\begin{split}
F^{(2)}_t = F_t-F^{(1)}_t
=&\; \int_0^t[(\bar{A}^{1/2})^{m+1,\cdot}-(w-\hat{w})'A^{-1}(\bar{A}^{1/2})^{x,\cdot}](X_s,s)dB_s \\
&\;+\int_0^t\left[\bar{w}^{m+1}-\frac{1}{2}(w-\hat{w})'A^{-1}(w-\hat{w})\right](X_s,s)ds.
\end{split}
\end{equation*}
We next make a crucial observation that
\begin{equation*}
\bar{A}^{m+1,m+1} = 2\bar{w}^{m+1},\;\;\;\bar{A}^{m+1,x} = w-\hat{w},
\end{equation*}
which shows that $e^{-F_t}$ is an exponential martingale. Moreover, straightforward calculations show that
\begin{equation*}
\begin{split}
&\;[(\bar{A}^{1/2})^{m+1,\cdot}-(w-\hat{w})'A^{-1}(\bar{A}^{1/2})^{x,\cdot}]
[(\bar{A}^{1/2})^{m+1,\cdot}-(w-\hat{w})'A^{-1}(\bar{A}^{1/2})^{x,\cdot}]' \\
=&\; \bar{A}^{m+1,m+1}-2(\bar{A}^{1/2})^{m+1,\cdot}(\bar{A}^{1/2})^{\cdot,x}A^{-1}(w-\hat{w})
+(w-\hat{w})'A^{-1}(w-\hat{w}) \\
=&\; 2\bar{w}^{m+1}-(w-\hat{w})'A^{-1}(w-\hat{w}).
\end{split}
\end{equation*}
This shows that $e^{-F^{(2)}_t}$ is also an exponential martingale. Since $M_t = e^{-F^{(1)}_t}$ and $N_t = e^{-F^{(2)}_t}$ are both martingales and $M_tN_t = e^{-F_t}$ is a martingale, we must have $[M,N]_t = 0$ and thus $M$ and $N$ are orthogonal martingales.
\end{proof}

\section{Proof of Theorem \ref{backward}}\label{backwardproof}
The proof of Theorem \ref{backward} is much more complicated than that of Theorem \ref{forward}. To do this, we consider an auxiliary process $(\bar{X}^\epsilon,\bar{Y}^\epsilon) = (\bar{X}^\epsilon_t,\bar{Y}^\epsilon_t)_{0\leq t\leq T}$ solving the SDE
\begin{equation*}
\begin{cases}
d\bar{X}^\epsilon_t
= \frac{1}{2}(\partial_{x_j}a^{ij}+\epsilon^{-1}\partial_{y_j}h^{ij})
(\bar{X}^\epsilon_t,\bar{Y}^\epsilon_t,t)dt
+\sigma(\bar{X}^\epsilon_t,\bar{Y}^\epsilon_t,t)dW_t,\\
d\bar{Y}^\epsilon_t
= \frac{1}{2}(\epsilon^{-1}\partial_{x_j}h^{ji}+\frac{1}{2}\partial_{y_j}\alpha^{ij})
(\bar{X}^\epsilon_t,\bar{Y}^\epsilon_t,t)dt
+\epsilon^{-1}\eta(\bar{X}^\epsilon_t,\bar{Y}^\epsilon_t,t)dW_t,
\end{cases}
\end{equation*}
with its initial distribution being chosen as that of $(X^\epsilon,Y^\epsilon)$. Then $(\bar{X}^\epsilon,\bar{Y}^\epsilon)$ is a diffusion process with drift
\begin{equation*}
\bar{B} = \frac{1}{2}\begin{pmatrix}
\nabla_x\cdot a+\epsilon^{-1}\nabla_y\cdot h \\
\epsilon^{-1}\nabla_x\cdot h'+\epsilon^{-2}\nabla_y\cdot\alpha\end{pmatrix}.
\end{equation*}
and diffusion matrix $D$. Moreover, we consider another process $(\bar{X}^{R,\epsilon},\bar{Y}^{R,\epsilon}) = (\bar{X}^{R,\epsilon}_t,\bar{Y}^{R,\epsilon}_t)_{0\leq t\leq T}$ solving the SDE
\begin{equation*}
\begin{cases}
d\bar{X}^{R,\epsilon}_t
= \frac{1}{2}(\partial_{x_j}a^{ij}+\epsilon^{-1}\partial_{y_j}h^{ij})
(\bar{X}^{R,\epsilon}_t,\bar{Y}^{R,\epsilon}_t,T-t)dt
+\sigma(\bar{X}^{R,\epsilon}_t,\bar{Y}^{R,\epsilon}_t,t)dW_t,\\
d\bar{Y}^{R,\epsilon}_t
= \frac{1}{2}(\epsilon^{-1}\partial_{x_j}h^{ji}+\frac{1}{2}\partial_{y_j}\alpha^{ij})
(\bar{X}^{R,\epsilon}_t,\bar{Y}^{R,\epsilon}_t,T-t)dt
+\epsilon^{-1}\eta(\bar{X}^{R,\epsilon}_t,\bar{Y}^{R,\epsilon}_t,t)dW_t,
\end{cases}
\end{equation*}
with its initial distribution being chosen as that of $(\tilde{X}^{R,\epsilon},\tilde{Y}^{R,\epsilon})$. Then $(\bar{X}^{R,\epsilon},\bar{Y}^{R,\epsilon})$ is a diffusion process with drift $\bar{B}(x,y,T-t)$ and diffusion matrix $D(x,y,T-t)$.

On the other hand, $(\tilde{X}^{R,\epsilon},\tilde{Y}^{R,\epsilon})$ is a diffusion process with drift
$\tilde{B}(x,y,T-t)$ and the same diffusion matrix $D(x,y,T-t)$, where
\begin{equation*}
\tilde{B} = \begin{pmatrix}
\tilde{b}+\epsilon^{-1}\tilde{f} \\
\epsilon^{-1}\tilde{g}+\epsilon^{-2}c\end{pmatrix}.
\end{equation*}
Let $\bar{P}^\epsilon_T$ be the law of $(\bar{X}^\epsilon,\bar{Y}^\epsilon)$, let $\bar{P}^{R,\epsilon}_T$ be the law of $(\bar{X}^{R,\epsilon},\bar{Y}^{R,\epsilon})$, and let $\bar{Q}^{R,\epsilon}_T$ be the law of its time reversal $(\bar{X}^{R,\epsilon}_{T-t},\bar{Y}^{R,\epsilon}_{T-t})_{0\leq t\leq T}$. Then the backward thermodynamic functional $G^\epsilon_T$ can be decomposed into three parts as
\begin{equation*}
G^\epsilon_T = \log\frac{dP^\epsilon_T}{d\bar{P}^\epsilon_T}(X^\epsilon_\cdot,Y^\epsilon_\cdot)
+\log\frac{d\bar{P}^\epsilon_T}{d\bar{Q}^{R,\epsilon}_T}(X^\epsilon_\cdot,Y^\epsilon_\cdot)
+\log\frac{d\bar{Q}^{R,\epsilon}_T}{d\tilde{Q}^{R,\epsilon}_T}(X^\epsilon_\cdot,Y^\epsilon_\cdot)
= \mathrm{I}+\mathrm{II}+\mathrm{III}.
\end{equation*}
We next calculate the explicit expressions of the three parts.

\begin{lemma}\label{part1}
\begin{equation*}
\mathrm{I} = \int_0^T(B-\bar{B})'D^{-1}\Sigma(X^\epsilon_s,Y^\epsilon_s,s)dW_s
+\frac{1}{2}\int_0^T(B-\bar{B})'D^{-1}(B-\bar{B})(X^\epsilon_s,Y^\epsilon_s,s)ds.
\end{equation*}
\end{lemma}

\begin{proof}
Since $(X^\epsilon,Y^\epsilon)$ and $(\bar{X}^\epsilon,\bar{Y}^\epsilon)$ have the same diffusion matrix, it follows from Girsanov's theorem that their laws are absolutely continuous with respect to each other and
\begin{equation*}
\frac{dP^\epsilon_T}{d\bar{P}^\epsilon_T}(X^\epsilon_\cdot,Y^\epsilon_\cdot)
= e^{\int_0^T(B-\bar{B})'D^{-1}\Sigma(X^\epsilon_s,Y^\epsilon_s,s)dW_s
+\frac{1}{2}\int_0^T(B-\bar{B})'D^{-1}(B-\bar{B})(X^\epsilon_s,Y^\epsilon_s,s)ds},
\end{equation*}
which completes the proof.
\end{proof}

\begin{lemma}
For each $t$, let $p^\epsilon(x,y,t)$ be the probability density of $(X^\epsilon_t,Y^\epsilon_t)$. Then
\begin{equation*}
\mathrm{II} = \log\frac{p^\epsilon(X^\epsilon_0,Y^\epsilon_0,0)}{p^\epsilon(X^\epsilon_T,Y^\epsilon_T,T)}.
\end{equation*}
\end{lemma}

\begin{proof}
Let $\bar{p}(x,y,t|x',y',s)$ be the transition probability density of the auxiliary process $(\bar{X}^\epsilon,\bar{Y}^\epsilon)$ from time $s$ to time $t$. Then it satisfies the Kolmogorov forward equation
\begin{equation*}
\partial_t\bar{p}(x,y,t|x',y',s) = (\bar{L}_t)^*\bar{p}(x,y,t|x',y',s),
\end{equation*}
where $\bar{L}_t$ is the generator of $(\bar{X}^\epsilon,\bar{Y}^\epsilon)$ and $(\bar{L}_t)^*$ is the adjoint operator of $L_t$ with respect to the Lebesgue measure. By the definition of the auxiliary process, it is easy to check that $\bar{L}_t$ is symmetric with respect to the Lebesgue measure. Therefore, we have
\begin{equation*}
\partial_t\bar{p}(x,y,t|x',y',s) = \bar{L}_t\bar{p}(x,y,t|x',y',s).
\end{equation*}
On the other hand, let $\bar{p}^R(x',y',t|x,y,s)$ be the transition probability density of the process $(\bar{X}^{R,\epsilon},\bar{Y}^{R,\epsilon})$ from time $s$ to time $t$. Then it satisfies the Kolmogorov backward equation
\begin{equation*}
\partial_s\bar{p}^R(x',y',t|x,y,s) = -\bar{L}^R_s\bar{p}^R(x',y',t|x,y,s),
\end{equation*}
where $\bar{L}^R_s$ is the generator of $(\bar{X}^{R,\epsilon},\bar{Y}^{R,\epsilon})$. This shows that
\begin{equation*}
\begin{split}
\partial_t\bar{p}^R(x',y',T-s|x,y,T-t) &= \bar{L}^R_{T-t}\bar{p}^R(x',y',T-s|x,y,T-t) \\
&= \bar{L}_t\bar{p}^R(x',y',T-s|x,y,T-t),
\end{split}
\end{equation*}
where have used the fact that $\bar{L}^R_{T-t} = \bar{L}_t$. Therefore, both the functions $f(x,y,t) = \bar{p}(x,y,t|x',y',s)$ and $g(x,y,t) = \bar{p}^R(x',y',T-s|x,y,T-t)$ satisfy the following Cauchy problem of the parabolic equation:
\begin{equation*}
\begin{cases}
\partial_tu = \bar{L}_tu,\;\;\;s\leq t\leq T,\\
u(s) = \delta(x-x')\delta(y-y').
\end{cases}
\end{equation*}
By the uniqueness of this parabolic equation, we have
\begin{equation}\label{symmetry}
\bar{p}(x,y,t|x',y',s) = \bar{p}^R(x',y',T-s|x,y,T-t).
\end{equation}
To proceed, we consider the finite-dimensional measurable cylinder set
\begin{equation*}
\begin{split}
E =&\; \{(x,y)\in C([0,T],\Rnum^{m+n}):
x(t_0)\in A_0,y(t_0)\in B_0,x(t_1)\in A_1,y(t_1)\in B_1,\cdots, \\
&\; x(t_N)\in A_N,\cdots,y(t_N)\in B_N,0=t_0<t_1<\cdots<t_N=T\},
\end{split}
\end{equation*}
where $A_i$ are Borel sets in $\Rnum^m$ and $B_i$ are Borel sets in $\Rnum^n$. Then we have
\begin{equation*}
\bar{P}^\epsilon_T(E) = \int_{A_0\times B_0\times\cdots A_N\times B_N}
\Pnum(\bar{X}^\epsilon_{t_0}\in dx_0,\bar{Y}^\epsilon_{t_0}\in dy_0,\cdots
\bar{X}^\epsilon_{t_N}\in dx_N,\bar{Y}^\epsilon_{t_N}\in dy_N).
\end{equation*}
For each $t$, let $\bar{p}(x,y,t)$ be the probability density of $(\bar{X}^\epsilon_t,\bar{Y}^\epsilon_t)$ and let $\bar{p}^R(x,y,t)$ be the probability density of $(\bar{X}^{R,\epsilon}_t,\bar{Y}^{R,\epsilon}_t)$. It then follows from \eqref{symmetry} that
\begin{equation*}
\begin{split}
&\; \Pnum(\bar{X}^\epsilon_{t_0}\in dx_0,\bar{Y}^\epsilon_{t_0}\in dy_0,\cdots
\bar{X}^\epsilon_{t_N}\in dx_N,\bar{Y}^\epsilon_{t_N}\in dy_N) \\
=&\; \bar{p}(x_0,y_0,t_0)\bar{p}(x_1,y_1,t_1|x_0,y_0,t_0)\cdots
\bar{p}(x_N,y_N,t_N|x_{N-1},y_{N-1},t_{N-1})dx_0dy_0\cdots dx_Ndy_N \\
=&\; \bar{p}(x_0,y_0,0)\bar{p}^R(x_0,y_0,T-t_0|x_1,y_1,T-t_1)\cdots
\bar{p}^R(x_{N-1},y_{N-1},T-t_{N-1}|x_N,y_N,T-t_N)\\
&\;dx_0dy_0\cdots dx_Ndy_N \\
=&\; \frac{\bar{p}(x_0,y_0,0)}{\bar{p}^R(x_N,y_N,0)}
\Pnum(\bar{X}^{R,\epsilon}_{T-t_N}\in dx_N,\bar{Y}^{R,\epsilon}_{T-t_N}\in dy_N,\cdots
\bar{X}^{R,\epsilon}_{T-t_0}\in dx_0,\bar{Y}^{R,\epsilon}_{T-t_0}\in dy_0).
\end{split}
\end{equation*}
Since $(\bar{X}^\epsilon,\bar{Y}^\epsilon)$ and $(X^\epsilon,Y^\epsilon)$ have the same initial distribution and $(\bar{X}^{R,\epsilon},\bar{Y}^{R,\epsilon})$ and $(\tilde{X}^{R,\epsilon},\tilde{Y}^{R,\epsilon})$ have the same initial distribution, which is the finial distribution of $(X^\epsilon,Y^\epsilon)$, we have
\begin{equation*}
\bar{p}(x_0,y_0,0) = p^\epsilon(x_0,y_0,0),\;\;\;
\bar{p}^R(x_N,y_N,0) = p^\epsilon(x_N,y_N,T).
\end{equation*}
Therefore, we obtain
\begin{equation*}
\begin{split}
\bar{P}^\epsilon_T(E) =&\; \int\Pnum(\bar{X}^\epsilon_{t_0}\in dx_0,\bar{Y}^\epsilon_{t_0}\in dy_0,\cdots
\bar{X}^\epsilon_{t_N}\in dx_N,\bar{Y}^\epsilon_{t_N}\in dy_N) \\
=&\; \int\frac{p^\epsilon(x_0,y_0,0)}{p^\epsilon(x_N,y_N,T)}
\Pnum(\bar{X}^{R,\epsilon}_{T-t_0}\in dx_0,\bar{Y}^{R,\epsilon}_{T-t_0}\in dy_0,\cdots,
\bar{X}^{R,\epsilon}_{T-t_N}\in dx_N,\bar{Y}^{R,\epsilon}_{T-t_N}\in dy_N) \\
=&\; \int_E\frac{p^\epsilon(x(0),y(0),0)}{p^\epsilon(x(T),y(T),T)}d\bar{Q}^{R,\epsilon}(x,y),
\end{split}
\end{equation*}
where first two integrals are taken over $A_0\times B_0\times\cdots A_N\times B_N$. Since the path space $C([0,T],\Rnum^{n+m})$ can be generated by finite-dimensional measurable cylinder sets, we finally obtain
\begin{equation*}
\frac{d\bar{P}^\epsilon_T}{d\bar{Q}^{R,\epsilon}_T}(X^\epsilon_\cdot,Y^\epsilon_\cdot)
= \frac{p^\epsilon(X^\epsilon_0,Y^\epsilon_0,0)}{p^\epsilon(X^\epsilon_T,X^\epsilon_T,T)},
\end{equation*}
which gives the desired result.
\end{proof}

\begin{lemma}\label{part3}
\begin{equation*}
\begin{split}
\mathrm{III} =&\; \int_0^T(\tilde{B}-\bar{B})'D^{-1}\Sigma(X^\epsilon_s,Y^\epsilon_s,s)dW_s
+\frac{1}{2}\int_0^T(\tilde{B}-\bar{B})'D^{-1}(\tilde{B}-\bar{B})(X^\epsilon_s,Y^\epsilon_s,s)ds \\
&\; +\int_0^T(B-\bar{B})'D^{-1}(\tilde{B}-\bar{B})(X^\epsilon_s,Y^\epsilon_s,s)ds
+\int_0^T\nabla\cdot(\tilde{B}-\bar{B})(X^\epsilon_s,Y^\epsilon_s,s)ds,
\end{split}
\end{equation*}
\end{lemma}

\begin{proof}
Since $(\bar{X}^{R,\epsilon},\bar{Y}^{R,\epsilon})$ and $(\tilde{X}^{R,\epsilon},\tilde{Y}^{R,\epsilon})$ have the same diffusion matrix, it follows from Girsanov's theorem that their laws are absolutely continuous with respect to each other and
\begin{equation*}
\frac{d\bar{P}^{R,\epsilon}_T}{d\tilde{P}^{R,\epsilon}_T}(\bar{X}^\epsilon_\cdot,\bar{Y}^\epsilon_\cdot)
= e^{\int_0^T(\bar{B}-\tilde{B})'D^{-1}\Sigma(\bar{X}^\epsilon_s,\bar{Y}^\epsilon_s,T-s)dW_s
+\frac{1}{2}\int_0^T(\bar{B}-\tilde{B})'D^{-1}(\bar{B}-\tilde{B})(\bar{X}^\epsilon_s,\bar{Y}^\epsilon_s,T-s)ds}.
\end{equation*}
For convenience of notation, set $\bar{Z}^{R,\epsilon} = (\bar{X}^{R,\epsilon},\bar{Y}^{R,\epsilon})'$. Then we have
\begin{equation*}
d\bar{Z}^{R,\epsilon}_t = \bar{B}(\bar{Z}^{R,\epsilon}_t,T-t)dt+\Sigma(\bar{Z}^{R,\epsilon}_t,T-t)dW_t.
\end{equation*}
This shows that
\begin{equation*}
\frac{d\bar{P}^{R,\epsilon}_T}{d\tilde{P}^{R,\epsilon}_T}(\bar{Z}^\epsilon_\cdot)
= e^{\int_0^T(\bar{B}-\tilde{B})'D^{-1}(\bar{Z}^{R,\epsilon}_s,T-s)d\bar{Z}^{R,\epsilon}_s
-\frac{1}{2}\int_0^T(\bar{B}-\tilde{B})'D^{-1}(\bar{B}+\tilde{B})(\bar{Z}^{R,\epsilon}_s,T-s)ds}.
\end{equation*}
In other words, we have
\begin{equation*}
\begin{split}
\log\frac{d\bar{P}^{R,\epsilon}_T}{d\tilde{P}^{R,\epsilon}_T}(w)
= \int_0^T(\bar{B}-\tilde{B})'D^{-1}(w(s),T-s)dw(s)-\int_0^TC(w(s),T-s)ds := M(w),
\end{split}
\end{equation*}
where
\begin{equation*}
C = \frac{1}{2}(\bar{B}-\tilde{B})'D^{-1}(\bar{B}+\tilde{B}).
\end{equation*}
Let $\phi:C([0,T],\Rnum^{m+n})\rightarrow C([0,T],\Rnum^{m+n})$ be the time-reversed operator on the path space defined as
\begin{equation*}
\phi(w(\cdot)) = w(T-\cdot).
\end{equation*}
Clearly, we have $\phi_*\bar{P}^{R,\epsilon}_T = \bar{Q}^{R,\epsilon}_T$ and $\phi_*\tilde{P}^{R,\epsilon}_T = \tilde{Q}^{R,\epsilon}_T$, where $\phi_*\mu$ is the push-forward measure of $\mu$ defined as $\phi_*\mu(\cdot) = \mu(\phi^{-1}(\cdot))$. It thus follows that for any Borel set $E\subset C([0,T],\Rnum^{m+n})$,
\begin{equation*}
\bar{Q}^{R,\epsilon}_T(E) = \bar{P}^{R,\epsilon}_T(\phi^{-1}(E))
= \int_{\phi^{-1}(E)}e^{M(w)}d\tilde{P}^{R,\epsilon}_T
= \int_{E}e^{M(\phi(w))}d\phi_*\tilde{P}^{R,\epsilon}_T
= \int_{E}e^{M(\phi(w))}d\tilde{Q}^{R,\epsilon}_T.
\end{equation*}
This shows that
\begin{equation*}
\frac{d\bar{Q}^{R,\epsilon}_T}{d\tilde{Q}^{R,\epsilon}_T}(\cdot) = e^{M(\phi(\cdot))}.
\end{equation*}
For convenience of notation, set $Z^\epsilon = (X^\epsilon,Y^\epsilon)'$. Then we have
\begin{equation*}
\begin{split}
\log\frac{d\bar{Q}^{R,\epsilon}_T}{d\tilde{Q}^{R,\epsilon}_T}(Z^\epsilon_\cdot)
&= \int_0^T(\bar{B}-\tilde{B})'D^{-1}(Z^\epsilon_{T-s},T-s)dZ^\epsilon_{T-s}
-\int_0^TC(Z^\epsilon_{T-s},T-s)ds \\
&= \int_0^T(\tilde{B}-\bar{B})'D^{-1}(Z^\epsilon_s,s)*dZ^\epsilon_s-\int_0^TC(Z^\epsilon_s,s)ds,
\end{split}
\end{equation*}
where the first integral in the last equality is the right-endpoint stochastic integral (recall that Ito's integral is the left-endpoint stochastic integral and the Stratonovich's integral is the middle-endpoint stochastic integral). Since the sum of the right-endpoint and left-endpoint stochastic integrals is twice the middle-endpoint stochastic integral, we have $f*dZ = fdZ+dfdZ$, which implies that
\begin{equation*}
\begin{split}
\log\frac{d\bar{Q}^{R,\epsilon}_T}{d\tilde{Q}^{R,\epsilon}_T}(Z^\epsilon_\cdot)
=&\; \int_0^T(\tilde{B}-\bar{B})'D^{-1}(Z^\epsilon_s,s)dZ^\epsilon_s
+\int_0^Td(\tilde{B}-\bar{B})'D^{-1}(Z^\epsilon_s,s)dZ^\epsilon_s-\int_0^TC(Z^\epsilon_s,s)ds \\
=&\; \int_0^T(\tilde{B}-\bar{B})'D^{-1}\Sigma(Z^\epsilon_s,s)dW_s
+\int_0^T(\tilde{B}-\bar{B})'D^{-1}B(Z^\epsilon_s,s)ds\\
&\; +\int_0^T\partial_j[(\tilde{B}-\bar{B})'D^{-1}]_iD^{ij}(Z^\epsilon_s,s)ds
-\int_0^TC(Z^\epsilon_s,s)ds.
\end{split}
\end{equation*}
By the definition of the auxiliary process, it is easy to see that $\nabla\cdot D = 2\bar{B}$, which implies that
\begin{equation*}
\partial_j[(\tilde{B}-\bar{B})'D^{-1}]_iD^{ij}
= \nabla\cdot(\tilde{B}-\bar{B})-2(\tilde{B}-\bar{B})'D^{-1}\bar{B}.
\end{equation*}
Therefore, we obtain
\begin{equation*}
\begin{split}
\log\frac{d\bar{Q}^{R,\epsilon}_T}{d\tilde{Q}^{R,\epsilon}_T}(Z^\epsilon_\cdot)
=&\; \int_0^T(\tilde{B}-\bar{B})'D^{-1}\Sigma(Z^\epsilon_s,s)dW_s
+\frac{1}{2}\int_0^T(\tilde{B}-\bar{B})'D^{-1}(\tilde{B}-\bar{B})(Z^\epsilon_s,s)ds \\
&\; +\int_0^T(B-\bar{B})'D^{-1}(\tilde{B}-\bar{B})(Z^\epsilon_s,s)ds
+\int_0^T\nabla\cdot(\tilde{B}-\bar{B})(Z^\epsilon_s,s)ds.
\end{split}
\end{equation*}
which gives the desired result.
\end{proof}

So far, we have obtained the explicit expressions of the three parts. Combining these three parts, we obtain the following corollary.

\begin{corollary}\label{Gexpression}
\begin{equation*}
\begin{split}
G^\epsilon_T
=&\;\log\frac{p^\epsilon(X^\epsilon_0,Y^\epsilon_0,0)}{p^\epsilon(X^\epsilon_T,Y^\epsilon_T,T)}
+\int_0^T(B+\tilde{B}-2\bar{B})'D^{-1}\Sigma(X^\epsilon_s,Y^\epsilon_s,s)dW_s\\
&\;+\frac{1}{2}\int_0^T(B+\tilde{B}-2\bar{B})'D^{-1}(B+\tilde{B}-2\bar{B})(X^\epsilon_s,Y^\epsilon_s,s)ds
+\int_0^T\nabla\cdot(\tilde{B}-\bar{B})(X^\epsilon_s,Y^\epsilon_s,s)ds.
\end{split}
\end{equation*}
\end{corollary}

To proceed, note that the second part can be rewritten as
\begin{equation*}
\mathrm{II} = \log\frac{p^\epsilon(X^\epsilon_0,Y^\epsilon_0,0)}{\rho(X^\epsilon_0,Y^\epsilon_0,0)}
-\log\frac{p^\epsilon(X^\epsilon_T,Y^\epsilon_T,T)}{\rho(X^\epsilon_T,Y^\epsilon_T,T)}
+\log\frac{\rho(X^\epsilon_0,Y^\epsilon_0,0)}{\rho(X^\epsilon_T,Y^\epsilon_T,T)},
\end{equation*}
where $\rho$ is the pseudo-stationary density of $Y^x$ defined in \eqref{rho}. Therefore, the backward thermodynamic functional can be represented as
\begin{equation*}
G^\epsilon_T = I^\epsilon_T+H^\epsilon_T,
\end{equation*}
where
\begin{equation*}
I^\epsilon_T = \log\frac{p^\epsilon(X^\epsilon_0,Y^\epsilon_0,0)}{\rho(X^\epsilon_0,Y^\epsilon_0,0)}
-\log\frac{p^\epsilon(X^\epsilon_T,Y^\epsilon_T,T)}{\rho(X^\epsilon_T,Y^\epsilon_T,T)},\;\;\;
H^\epsilon_T = \mathrm{I}+\mathrm{III}+\log\frac{\rho(X^\epsilon_0,Y^\epsilon_0,0)}{\rho(X^\epsilon_T,Y^\epsilon_T,T)}.
\end{equation*}

We shall next investigate the weak limits of $H^\epsilon_T$ and $I^\epsilon_T$ separately. The following lemma gives another explicit expression of $H^\epsilon_T$, whose weak limit is much easier to deal with.

\begin{lemma}\label{Hexpression}
\begin{equation*}
H^\epsilon_T = \int_0^T\sigma^{m+1}(X^\epsilon_s,Y^\epsilon_s,s)dW_s
+\int_0^T(b^{m+1}+\epsilon^{-1}f^{m+1})(X^\epsilon_s,Y^\epsilon_s,s)ds,
\end{equation*}
where
\begin{equation*}
\begin{split}
\sigma^{m+1} &= (B+\tilde{B}-2\bar{B}-D\nabla\log\rho)'D^{-1}\Sigma, \\
b^{m+1} &= \tfrac{1}{2}\sigma^{m+1}(\sigma^{m+1})'
+\tfrac{1}{\rho}\nabla_x\cdot[(\tilde{b}-\tfrac{1}{2}\nabla_x\cdot a-\tfrac{1}{2}a\nabla_x\log\rho)\rho]
-\partial_s\log\rho, \\
f^{m+1} &= \tfrac{1}{\rho}[\nabla_y\cdot(\tilde{g}\rho)-\nabla_x\cdot(f\rho)].
\end{split}
\end{equation*}
\end{lemma}

\begin{proof}
By Ito's formula, we have
\begin{equation*}
\begin{split}
\log\frac{\rho(X^\epsilon_0,Y^\epsilon_0,0)}{\rho(X^\epsilon_T,Y^\epsilon_T,T)}
=&\; -\int_0^Td\log\rho(X^\epsilon_s,Y^\epsilon_s,s)\\
=&\; -\int_0^T(\nabla\log\rho)'\Sigma(X^\epsilon_s,Y^\epsilon_s,s)dW_s \\
&\; -\int_0^T[(\nabla\log\rho)'B+\partial_s\log\rho +\tfrac{1}{2}\partial_{ij}\log\rho D^{ij}](X^\epsilon_s,Y^\epsilon_s,s)ds.
\end{split}
\end{equation*}
Since $\nabla\cdot D = 2\bar{B}$, we have
\begin{equation*}
\frac{1}{2}\partial_{ij}\log\rho D^{ij}
= \frac{1}{2}\nabla\cdot(D\nabla\log\rho)-\bar{B}'\nabla\log\rho.
\end{equation*}
This shows that
\begin{equation*}
\begin{split}
\log\frac{\rho(X^\epsilon_0,Y^\epsilon_0,0)}{\rho(X^\epsilon_T,Y^\epsilon_T,T)}
=&\; -\int_0^T(\nabla\log\rho)'\Sigma(X^\epsilon_s,Y^\epsilon_s,s)dW_s
-\int_0^T(B-\bar{B})'\nabla\log\rho(X^\epsilon_s,Y^\epsilon_s,s)ds \\
&\; -\frac{1}{2}\int_0^T\nabla\cdot(D\nabla\log\rho)(X^\epsilon_s,Y^\epsilon_s,s)ds
-\int_0^T\partial_s\log\rho(X^\epsilon_s,Y^\epsilon_s,s)ds.
\end{split}
\end{equation*}
This equation, together with Corollary \ref{Gexpression}, shows that
\begin{equation}\label{Hraw}
\begin{split}
H^\epsilon_T
=&\;\int_0^T(B+\tilde{B}-2\bar{B}-D\nabla\log\rho)'D^{-1}\Sigma(X^\epsilon_s,Y^\epsilon_s,s)dW_s\\
&\;+\frac{1}{2}\int_0^T(B+\tilde{B}-2\bar{B}-D\nabla\log\rho)'D^{-1}
(B+\tilde{B}-2\bar{B}-D\nabla\log\rho)(X^\epsilon_s,Y^\epsilon_s,s)ds \\
&\;+\int_0^T(\tilde{B}-\bar{B}-\tfrac{1}{2}D\nabla\log\rho)'\nabla\log\rho(X^\epsilon_s,Y^\epsilon_s,s)ds\\
&\;+\int_0^T\nabla\cdot(\tilde{B}-\bar{B}-\tfrac{1}{2}D\nabla\log\rho)(X^\epsilon_s,Y^\epsilon_s,s)ds
-\int_0^T\partial_s\log\rho(X^\epsilon_s,Y^\epsilon_s,s)ds.
\end{split}
\end{equation}
It is easy to see that
\begin{equation*}
\begin{split}
(\tilde{B}-\bar{B}-\tfrac{1}{2}D\nabla\log\rho)'\nabla\log\rho
+\nabla\cdot(\tilde{B}-\bar{B}-\tfrac{1}{2}D\nabla\log\rho)
= \tfrac{1}{\rho}\nabla\cdot[(\tilde{B}-\bar{B}-\tfrac{1}{2}D\nabla\log\rho)\rho].
\end{split}
\end{equation*}
Inserting this equation into \eqref{Hraw} yields
\begin{equation}\label{h1}
\begin{split}
H^\epsilon_T
=&\;\int_0^T(B+\tilde{B}-2\bar{B}-D\nabla\log\rho)'D^{-1}\Sigma(X^\epsilon_s,Y^\epsilon_s,s)dW_s\\
&\;+\frac{1}{2}\int_0^T(B+\tilde{B}-2\bar{B}-D\nabla\log\rho)'D^{-1}
(B+\tilde{B}-2\bar{B}-D\nabla\log\rho)(X^\epsilon_s,Y^\epsilon_s,s)ds \\
&\;+\int_0^T\tfrac{1}{\rho}\nabla\cdot[(\tilde{B}-\bar{B}-\tfrac{1}{2}D\nabla\log\rho)\rho](X^\epsilon_s,Y^\epsilon_s,s)ds
-\int_0^T\partial_s\log\rho(X^\epsilon_s,Y^\epsilon_s,s)ds.
\end{split}
\end{equation}
Moreover, it is easy to check that
\begin{equation}\label{matrix}
\begin{split}
&\;\tilde{B}-\bar{B}-\tfrac{1}{2}D\nabla\log\rho \\
=&\; \begin{pmatrix}
(\tilde{b}-\tfrac{1}{2}\nabla_x\cdot a-\tfrac{1}{2}a\nabla_x\log\rho)
+\epsilon^{-1}(\tilde{f}-\tfrac{1}{2}\nabla_y\cdot h-\tfrac{1}{2}h\nabla_y\log\rho) \\
\epsilon^{-1}(\tilde{g}-\tfrac{1}{2}\nabla_x\cdot h'-\tfrac{1}{2}h'\nabla_x\log\rho)
+\epsilon^{-2}(c-\tfrac{1}{2}\nabla_y\cdot\alpha-\tfrac{1}{2}\alpha\nabla_y\log\rho) \end{pmatrix}.
\end{split}
\end{equation}
By Assumption \ref{a2} and Lemma \ref{a2explain}, the pseudo-stationary density $\rho$ satisfies
\begin{equation*}
2c = \nabla_y\cdot\alpha+\alpha\nabla_y\log\rho.
\end{equation*}
Moreover, it follows from Assumption \ref{a1} that
\begin{equation*}
\tilde{f}-\tfrac{1}{2}\nabla_y\cdot h-\tfrac{1}{2}h\nabla_y\log\rho
= \tfrac{1}{2}\nabla_y\cdot h+\tfrac{1}{2}h\nabla_y\log\rho-f.
\end{equation*}
Therefore, the right side of \eqref{matrix} can be simplified as
\begin{equation*}
\tilde{B}-\bar{B}-\tfrac{1}{2}D\nabla\log\rho = \begin{pmatrix}
(\tilde{b}-\tfrac{1}{2}\nabla_x\cdot a-\tfrac{1}{2}a\nabla_x\log\rho)
+\epsilon^{-1}(\tfrac{1}{2}\nabla_y\cdot h+\tfrac{1}{2}h\nabla_y\log\rho-f) \\
\epsilon^{-1}(\tilde{g}-\tfrac{1}{2}\nabla_x\cdot h'-\tfrac{1}{2}h'\nabla_x\log\rho) \end{pmatrix}.
\end{equation*}
This shows that
\begin{equation*}
\begin{split}
&\; \nabla\cdot[(\tilde{B}-\bar{B}-\tfrac{1}{2}D\nabla\log\rho)\rho] \\
=&\; \nabla_x\cdot[(\tilde{b}-\tfrac{1}{2}\nabla_x\cdot a-\tfrac{1}{2}a\nabla_x\log\rho)\rho] \\
&\; +\epsilon^{-1}\nabla_x\cdot[(\tfrac{1}{2}\nabla_y\cdot h+\tfrac{1}{2}h\nabla_y\log\rho-f)\rho]
+\epsilon^{-1}\nabla_y\cdot[(\tilde{g}-\tfrac{1}{2}\nabla_x\cdot h'-\tfrac{1}{2}h'\nabla_x\log\rho)\rho] \\
=&\; \nabla_x\cdot[(\tilde{b}-\tfrac{1}{2}\nabla_x\cdot a-\tfrac{1}{2}a\nabla_x\log\rho)\rho]
+\epsilon^{-1}[\nabla_y\cdot(\tilde{g}\rho)-\nabla_x\cdot(f\rho)].
\end{split}
\end{equation*}
Inserting this into \eqref{h1} gives the desired result.
\end{proof}

The following lemma gives the limit of $I^\epsilon_T$.

\begin{lemma}\label{boundary}
For each $t$, let $p(x,t)$ be the probability density of $X_t$. Then for any $\eta>0$, we have
\begin{equation*}
\lim_{\epsilon\rightarrow 0}\Pnum\left(\sup_{t\leq T}
\left|I^\epsilon_t-\log\frac{p(X^\epsilon_0,0)}{p(X^\epsilon_t,t)}\right|>\eta\right) = 0.
\end{equation*}
\end{lemma}

\begin{proof}
We first assume that the original process $(X^\epsilon,Y^\epsilon)$ starts at $t = 0$. By the matched asymptotic expansions of singularly perturbed diffusion processes under second-order averaging \cite{khasminskii2005limit}, there exist two sequences of smooth functions $u_i(x,y,t)$ and $v_i(x,y,t)$ such that for any $n\geq 0$ and any compact set $K\subset\Rnum^{m+n}$,
\begin{equation*}
\sup_{(x,y,t)\in K\times[0,T]}\left|\sum_{i=0}^n\epsilon^iu_i(x,y,t)
+\sum_{i=0}^n\epsilon^iv_i(t/\epsilon,x,y)-p^\epsilon(x,y,t)\right| = O(\epsilon^{n+1}),
\end{equation*}
where $u_i(x,y,t)$ are called outer expansions and $v_i(x,y,t)$ are called initial layer corrections. Moreover, the initial layer corrections $v_i(x,y,t)$ decay exponentially fast:
\begin{equation*}
\sup_{(x,y)\in K}\left|v_i(x,y,t)\right| \leq c_1e^{-c_2t},
\end{equation*}
where $c_1$ and $c_2$ are two constants independent of $x$ and $y$. Although the above two estimates are proved in \cite{khasminskii2005limit} for singularly perturbed diffusions on compact Riemannian manifolds, the proof can be easily generalized to diffusions on the entire space \cite[Section 6]{khasminskii1996transition}. Combining the above two estimates, for any $\delta>0$, we have
\begin{equation*}
\begin{split}
&\; \sup_{(x,y,t)\in K\times[\delta,T]}|p^\epsilon(x,y,t)-u_0(x,y,t)| \\
\leq&\; \sup_{(x,y,t)\in K\times[0,T]}|p^\epsilon(x,y,t)-u_0(x,y,t)-v_0(t/\epsilon,x,y)|
+\sup_{(x,y,t)\in K\times[\delta,T]}|v_0(t/\epsilon,x,y)| \\
\leq&\; O(\epsilon^2)+c_1e^{-c_2\delta/\epsilon}.
\end{split}
\end{equation*}
Moreover, the outer expansion $u_0(x,y,t)$ is given by \cite{khasminskii1996transition}
\begin{equation*}
u_0(x,y,t) = p(x,t)\rho(x,y,t) > 0,
\end{equation*}
where the positivity of $u_0$ follows from Harnack's inequality \cite{herzog2015practical}. It thus follows from the above two equations that
\begin{equation*}
\lim_{\epsilon\rightarrow 0}\sup_{(x,y,t)\in K\times[\delta,T]}
\left|\frac{p^\epsilon(x,y,t)}{p(x,t)\rho(x,y,t)}-1\right| = 0.
\end{equation*}
Since we have assumed that that the original process $(X^\epsilon,Y^\epsilon)$ starts at $t = -1$, we can replace $\delta$ by $0$ in the above limit. This implies that
\begin{equation*}
\lim_{\epsilon\rightarrow 0}\sup_{(x,y,t)\in K\times[0,T]}
\left|\log\frac{p^\epsilon(x,y,t)}{\rho(x,y,t)p(x,t)}\right| = 0.
\end{equation*}
For any $\delta>0$, it follows from Assumption \ref{a3} that there exists a compact set $K_\delta\subset\Rnum^{m+n}$ such that
\begin{equation*}
\Pnum\left(\inf_{\epsilon>0}\tau^\epsilon_{K_\delta}<T\right) \leq \delta.
\end{equation*}
Therefore, for any $\eta > 0$, when $\epsilon$ is sufficiently small, we have
\begin{equation*}
\begin{split}
&\;\Pnum\left(\sup_{t\leq T}\left|\log\frac{p^\epsilon(X^\epsilon_t,Y^\epsilon_t,t)}
{\rho(X^\epsilon_t,Y^\epsilon_t,t)p(X^\epsilon_t,t)}\right|>\eta\right) \\
\leq&\; \Pnum\left(\sup_{t\leq T}\left|\log\frac{p^\epsilon(X^\epsilon_t,Y^\epsilon_t,t)}
{\rho(X^\epsilon_t,Y^\epsilon_t,t)p(X^\epsilon_t,t)}\right|>\eta,\;
\inf_{\epsilon>0}\tau^\epsilon_{K_\delta}\geq T\right)
+\Pnum\left(\inf_{\epsilon>0}\tau^\epsilon_{K_\delta}<T\right)
\leq \delta.
\end{split}
\end{equation*}
This finally shows that
\begin{equation*}
\lim_{\epsilon\rightarrow 0}
\Pnum\left(\sup_{t\leq T}\left|\log\frac{p^\epsilon(X^\epsilon_t,Y^\epsilon_t,t)}
{\rho(X^\epsilon_t,Y^\epsilon_t,t)p(X^\epsilon_t,t)}\right|>\eta\right) = 0,
\end{equation*}
which implies the desired result.
\end{proof}

We next derive the explicit expressions of $G_T$ and $G^{(1)}_T$. Clearly, we have
\begin{equation}\label{notdiverge}
\begin{split}
&\;B+\tilde{B}-2\bar{B}-D\nabla\log\rho \\
=&\; \begin{pmatrix} (b+\tilde{b}-\nabla_x\cdot a-a\nabla_x\log\rho)
+\epsilon^{-1}(f+\tilde{f}-\nabla_y\cdot h-h\nabla_y\log\rho) \\
\epsilon^{-1}(g+\tilde{g}-\nabla_x\cdot h'-h'\nabla_x\log\rho)
+\epsilon^{-2}(2c-\nabla_y\cdot\alpha-\alpha\nabla_y\log\rho) \end{pmatrix}.
\end{split}
\end{equation}
By Assumptions \ref{a1} and \ref{a2}, the right side of the above equation can be simplified as
\begin{equation*}
B+\tilde{B}-2\bar{B}-D\nabla\log\rho = \begin{pmatrix}
b+\tilde{b}-\nabla_x\cdot a-a\nabla_x\log\rho \\
\epsilon^{-1}(g+\tilde{g}-\nabla_x\cdot h'-h'\nabla_x\log\rho) \end{pmatrix}.
\end{equation*}
It is easy to check that the functions $\sigma^{m+1}$, $b^{m+1}$, and $f^{m+1}$ defined in Lemma \ref{Hexpression} are all independent of $\epsilon$. It then follows from Lemma \ref{Hexpression} that the ordered triple $(X^\epsilon,H^\epsilon,Y^\epsilon)$ is a diffusion process solving the SDE
\begin{equation*}\left\{
\begin{split}
dX^{\epsilon}_t
&= [b(X^\epsilon_t,Y^\epsilon_t,t)+\epsilon^{-1}f(X^\epsilon_t,Y^\epsilon_t,t)]dt
+\sigma(X^{\epsilon}_t,Y^{\epsilon}_t,t)dW_t,\\
dH^{\epsilon}_t
&= [b^{m+1}(X^\epsilon_t,Y^\epsilon_t,t)+\epsilon^{-1}f^{m+1}(X^\epsilon_t,Y^\epsilon_t,t)]dt
+\sigma^{m+1}(X^\epsilon_t,Y^\epsilon_t,t)dW_t,\\
dY^{\epsilon}_t
&= [\epsilon^{-1}g(X^\epsilon_t,Y^\epsilon_t,t)+\epsilon^{-2}c(X^\epsilon_t,Y^\epsilon_t,t)]dt
+\epsilon^{-1}\eta(X^\epsilon_t,Y^\epsilon_t,t)dW_t,
\end{split}\right.
\end{equation*}
where $X^\epsilon$ and $H^\epsilon$ are slowly varying and $Y^\epsilon$ is rapidly varying. This SDE can be rewritten as
\begin{equation*}\left\{
\begin{split}
&d\begin{pmatrix} X^\epsilon_t \\ H^\epsilon_t \end{pmatrix}
= [\bar{b}(X^\epsilon_t,Y^\epsilon_t,t)+\epsilon^{-1}\bar{f}(X^\epsilon_t,Y^\epsilon_t,t)]dt
+\bar\sigma(X^\epsilon_t,Y^\epsilon_t,t)dW_t,\\
&dY^{\epsilon}_t
= [\epsilon^{-1}g(X^\epsilon_t,Y^\epsilon_t,t)+\epsilon^{-2}c(X^\epsilon_t,Y^\epsilon_t,t)]dt
+\epsilon^{-1}\eta(X^\epsilon_t,Y^\epsilon_t,t)dW_t,
\end{split}\right.
\end{equation*}
where
\begin{equation*}
\bar{b} = \begin{pmatrix} b \\ b^{m+1} \end{pmatrix},\;\;\;
\bar{f} = \begin{pmatrix} f \\ f^{m+1} \end{pmatrix},\;\;\;
\bar{\sigma} = \begin{pmatrix} \sigma \\ \sigma^{m+1} \end{pmatrix}.
\end{equation*}
The diffusion matrix of $(X^\epsilon,H^\epsilon,Y^\epsilon)$ is given by
\begin{equation*}
\begin{pmatrix} \bar{\sigma} \\ \epsilon^{-1}\eta \end{pmatrix}
\begin{pmatrix} \bar{\sigma} \\ \epsilon^{-1}\eta \end{pmatrix}'
= \begin{pmatrix} \bar{a} & \epsilon^{-1}\bar{h} \\
\epsilon^{-1}\bar{h}' & \epsilon^{-2}\alpha \end{pmatrix},
\end{equation*}
where
\begin{gather*}
\bar{a} = \begin{pmatrix} a & b+\tilde{b}-\nabla_x\cdot a-a\nabla_x\log\rho \\
(b+\tilde{b}-\nabla_x\cdot a-a\nabla_x\log\rho)' & \sigma^{m+1}(\sigma^{m+1})' \end{pmatrix},\\
\bar{h} = \begin{pmatrix} h \\ (g+\tilde{g}-\nabla_x\cdot h'-h'\nabla_x\log\rho)' \end{pmatrix}.
\end{gather*}
Let $\bar\phi = (\phi,\phi^{m+1})'$ be solution to the Poisson equation
\begin{equation*}\left\{
\begin{split}
&-L_0\bar\phi(x,y,t) = \bar{f}(x,y,t), \\
&\int_{\Rnum^n}\bar\phi(x,y,t)\rho(x,y,t)dy = 0.
\end{split}\right.
\end{equation*}
It thus follows from Lemma \ref{convergence} that $(X^\epsilon,H^\epsilon)\Rightarrow(X,H)$ in $C([0,T],\Rnum^{m+1})$, where $(X,H)$ is a diffusion process with generator
\begin{equation*}
\bar{L} = \sum_{i=1}^{m+1}\bar{w}^i(x,t)\partial_{x_i}+\frac{1}{2}\sum_{i,j=1}^{m+1}\bar{A}^{ij}(x,t)\partial_{x_ix_j},
\end{equation*}
where
\begin{gather*}
\bar{w}^i(x,t) = \int_{\Rnum^n}(\bar{b}^i+\partial_{x_j}\bar\phi^if^j+\partial_{y_j}\bar\phi^ig^j
+\partial_{x_jy_k}\bar\phi^ih^{jk})(x,y,t)\rho(x,y,t)dy,\\
\bar{A}^{ij}(x,t) = \int_{\Rnum^n}(\bar{a}^{ij}+\bar\phi^i\bar{f}^j+\bar\phi^j\bar{f}^i
+\partial_{y_k}\bar\phi^i\bar{h}^{jk}+\partial_{y_k}\bar\phi^j\bar{h}^{ik})(x,y,t)\rho(x,y,t)dy.
\end{gather*}
It is easy to see that $\bar{w}^i(x,t) = w^i(x,t)$ and $\bar{A}^{ij}(x,t) = A^{ij}(x,t)$ for any $1\leq i,j\leq m$. Thus, $(X,H)$ can be viewed as the solution to the SDE
\begin{equation*}\left\{
\begin{split}
dX^i_t &= w^i(X_t,t)dt+(\bar{A}^{1/2})^{x,\cdot}(X_t,t)dB_t,\\
dH_t &= \bar{w}^{m+1}(X_t,t)dt+(\bar{A}^{1/2})^{m+1,\cdot}(X_t,t)dB_t,
\end{split}\right.
\end{equation*}
where $\bar{A} = (\bar{A}^{ij})$ is an $(m+1)\times(m+1)$ matrix, $(\bar{A}^{1/2})^{x,\cdot}$ is the first $m$ rows of $\bar{A}^{1/2}$, $(\bar{A}^{1/2})^{m+1,\cdot}$ is the last row of $\bar{A}^{1/2}$, and $B = (B_t)_{t\geq 0}$ is an $(m+1)$-dimensional standard Brownian motion defined on some probability space. Since the diffusion matrix of $X$ is $A(x,t) = (A^{ij}(x,t))$, we have $\bar{A}^{x,x} = A$, where $\bar{A}^{x,x}$ the matrix obtained from $\bar{A}$ by retaining the first $m$ rows and first $m$ columns. Let $\tilde{X}^R$ be the averaged process of the comparable process $(\tilde{X}^{R,\epsilon},\tilde{Y}^{R,\epsilon})$. By Proposition \ref{same}, the diffusion matrix of $\tilde{X}^R$ is $A(x,T-t)$ and thus $\tilde{X}^R$ can be viewed as the solution to the SDE
\begin{equation*}
d\tilde{X}^R_t = \tilde{w}(\tilde{X}^R_t,T-t)dt+(\bar{A}^{1/2})^{x,\cdot}(\tilde{X}^R_t,T-t)dB_t.
\end{equation*}
Imitating the proof of Corollary \ref{Gexpression}, we can also prove that the laws of $(X_t)_{0\leq t\leq T}$ and $(\tilde{X}^R_{T-t})_{0\leq t\leq T}$ are absolutely continuous with respect to each other and the thermodynamic functional $G^{(1)}_t$ has the following explicit expression:
\begin{equation}\label{G1}
\begin{split}
G^{(1)}_T =&\;\log\frac{p(X_0,0)}{p(X_T,T)}
+\int_0^T(w+\tilde{w}-\nabla_x\cdot A)'A^{-1}(\bar{A}^{1/2})^{x,\cdot}(X_s,s)dB_s\\
&\;+\frac{1}{2}\int_0^T(w+\tilde{w}-\nabla_x\cdot A)A^{-1}(w+\tilde{w}-\nabla_x\cdot A)(X_s,s)ds\\
&\;+\int_0^T\nabla\cdot(\tilde{w}-\tfrac{1}{2}\nabla_x\cdot A)(X_s,s)ds.
\end{split}
\end{equation}
Obviously, we have
\begin{equation*}
H_T = \int_0^T(\bar{A}^{1/2})^{m+1,\cdot}(X_s,s)dB_s+\int_0^T\bar{w}^{m+1}(X_s,s)ds.
\end{equation*}
Since $(X^\epsilon,H^\epsilon)\Rightarrow(X,H)$ in $C([0,T],\Rnum^{m+1})$, we have
\begin{equation*}
\left(X^\epsilon_t,\log\frac{p(X^\epsilon_0,0)}{p(X^\epsilon_t,t)}+H^\epsilon_t\right)
\Rightarrow \left(X_t,\log\frac{p(X_0,0)}{p(X_t,t)}+H_t\right).
\end{equation*}
in $C([0,T],\Rnum^{m+1})$ \cite[Page 350]{jacod2002limit}. Moreover, it follows from Lemma \ref{boundary} that for any $\eta>0$,
\begin{equation*}
\lim_{\epsilon\rightarrow 0}\Pnum\left(\sup_{t\leq T}\left|I^\epsilon_t
-\log\frac{p(X^\epsilon_0,0)}{p(X^\epsilon_t,t)}\right|>\eta\right) = 0.
\end{equation*}
It then follows from the above two equations and \cite[Page 352, Lemma 3.31]{jacod2002limit} that
\begin{equation*}
\left(X^\epsilon_t,G^\epsilon_t\right) = \left(X^\epsilon_t,I^\epsilon_t+H^\epsilon_t\right)
\Rightarrow \left(X_t,\log\frac{p(X_0,0)}{p(X_t,t)}+H_t\right)
\end{equation*}
in $C([0,T],\Rnum^{m+1})$. Therefore, the limit functional $G_t$ has the following explicit expression:
\begin{equation}\label{G}
G_t = \log\frac{p(X_0,0)}{p(X_t,t)}+H_t.
\end{equation}

The following three lemmas play an important role in proving that the anomalous part of the limit functional $G_t$ is a martingale.

\begin{lemma}\label{w1}
\begin{equation*}
\bar{w}^{m+1} = \frac{1}{2}\bar{A}^{m+1,m+1}+\nabla_x\cdot\int_{\Rnum^n}(\tilde{b}-\tfrac{1}{2}\nabla_x\cdot a-\tfrac{1}{2}a\nabla_x\log\rho-\phi^{m+1}\tilde{f})\rho dy.
\end{equation*}
\end{lemma}

\begin{proof}
On one hand, we have
\begin{equation*}
\begin{split}
&\;\int_{\Rnum^n}\bar{b}^{m+1}(x,y,t)\rho(x,y,t)dy\\
=&\;\frac{1}{2}\int_{\Rnum^n}\sigma^{m+1}(\sigma^{m+1})'\rho dy
+\int_{\Rnum^n}\nabla_x\cdot[(\tilde{b}-\tfrac{1}{2}\nabla_x\cdot a-\tfrac{1}{2}a\nabla_x\log\rho)\rho]dy
-\int_{\Rnum^n}\partial_s\rho dy\\
=&\;\frac{1}{2}\int_{\Rnum^n}\bar{a}^{m+1,m+1}\rho dy
+\nabla_x\cdot\int_{\Rnum^n}(\tilde{b}-\tfrac{1}{2}\nabla_x\cdot a-\tfrac{1}{2}a\nabla_x\log\rho)\rho dy.
\end{split}
\end{equation*}
On the other hand, it follows from Lemma \ref{Hexpression}, Assumption \ref{a1}, and integration by parts that
\begin{equation*}
\begin{split}
&\;\int_{\Rnum^n}[\partial_{x_j}\phi^{m+1}f^j+\partial_{y_j}\phi^{m+1}g^j
+\partial_{x_jy_k}\phi^{m+1}h^{jk}](x,y,t)\rho(x,y,t)dy\\
=&\;\partial_{x_j}\int_{\Rnum^n}\phi^{m+1}f^j\rho dy
-\int_{\Rnum^n}\phi^{m+1}[\nabla_x\cdot(f\rho)+\nabla_y\cdot(g\rho)]dy
+\int_{\Rnum^n}\partial_{x_jy_k}\phi^{m+1}h^{jk}\rho dy\\
=&\;\partial_{x_j}\int_{\Rnum^n}\phi^{m+1}f^j\rho dy
+\int_{\Rnum^n}\phi^{m+1}[\nabla_y\cdot(\tilde{g}\rho)-\nabla_x\cdot(f\rho)]dy
-\int_{\Rnum^n}\phi^{m+1}\nabla_y\cdot[(g+\tilde{g})\rho]dy\\
&\;+\partial_{x_j}\int_{\Rnum^n}\partial_{y_k}\phi^{m+1}h^{jk}\rho dy
-\int_{\Rnum^n}\partial_{y_k}\phi^{m+1}\partial_{x_j}(h^{jk}\rho)dy\\
=&\;\partial_{x_j}\int_{\Rnum^n}\phi^{m+1}f^j\rho dy
+\int_{\Rnum^n}\phi^{m+1}f^{m+1}\rho dy
+\int_{\Rnum^n}\partial_{y_k}\phi^{m+1}(g^k+\tilde{g}^k)\rho dy\\
&\;-\partial_{x_j}\int_{\Rnum^n}\phi^{m+1}\partial_{y_k}(h^{jk}\rho)dy
-\int_{\Rnum^n}\partial_{y_k}\phi^{m+1}\partial_{x_j}(h^{jk}\rho)dy\\
=&\;\partial_{x_j}\int_{\Rnum^n}\phi^{m+1}(f^j-\partial_{y_k}h^{jk}-h^{jk}\partial_{y_k}\log\rho)\rho dy
+\int_{\Rnum^n}\phi^{m+1}f^{m+1}\rho dy\\
&\;+\int_{\Rnum^n}\partial_{y_k}\phi^{m+1}(g^k+\tilde{g}^k
-\partial_{x_j}h^{jk}-h^{jk}\partial_{x_j}\log\rho)\rho dy\\
=&\;-\nabla_x\cdot\int_{\Rnum^n}\phi^{m+1}\tilde{f}\rho dy
+\int_{\Rnum^n}\phi^{m+1}f^{m+1}\rho dy+\int_{\Rnum^n}\partial_{y_k}\phi^{m+1}\bar{h}^{m+1,k}\rho dy\\
\end{split}
\end{equation*}
Combing the above two equations yields
\begin{equation*}
\begin{split}
\bar{w}^{m+1}(x,t)
=&\;\int_{\Rnum^n}(\bar{b}^{m+1}+\partial_{x_j}\phi^{m+1}f^j+\partial_{y_j}\phi^{m+1}g^j
+\partial_{x_jy_k}\phi^{m+1}h^{jk})\rho dy\\
=&\;\frac{1}{2}\int_{\Rnum^n}(\bar{a}^{m+1,m+1}+2\phi^{m+1}f^{m+1}+2\partial_{y_k}\phi^{m+1}\bar{h}^{m+1,k})\rho dy\\
&\;+\nabla_x\cdot\int_{\Rnum^n}(\tilde{b}-\tfrac{1}{2}\nabla_x\cdot a-\tfrac{1}{2}a\nabla_x\log\rho-\phi^{m+1}\tilde{f})\rho dy\\
=&\;\frac{1}{2}\bar{A}^{m+1,m+1}+\nabla_x\cdot\int_{\Rnum^n}(\tilde{b}-\tfrac{1}{2}\nabla_x\cdot a-\tfrac{1}{2}a\nabla_x\log\rho-\phi^{m+1}\tilde{f})\rho dy.
\end{split}
\end{equation*}
This completes the proof.
\end{proof}

\begin{lemma}\label{w2}
\begin{equation*}
\bar{A}^{m+1,x} = w+\tilde{w}-\nabla_x\cdot A.
\end{equation*}
\end{lemma}

\begin{proof}
For any $1\leq j\leq m$, we have
\begin{equation*}
\bar{A}^{m+1,j}(x,t)
= \int_{\Rnum^n}(\bar{a}^{m+1,j}+\phi^{m+1}f^j+\phi^jf^{m+1}
+\partial_{y_k}\phi^{m+1}h^{jk}+\partial_{y_k}\phi^j\bar{h}^{m+1,k})\rho dy.
\end{equation*}
It follows from Assumption \ref{a2} and Lemma \ref{a2explain} that $L_0$ is symmetric with respect to $\rho$. Therefore, for any $1\leq i,j\leq m+1$, we have
\begin{equation}\label{symmetric}
\begin{split}
\int_{\Rnum^n}\phi^if^j\rho dy = -(\phi^i,L_0\phi^j)_\rho
= -(L_0\phi^i,\phi^j)_\rho = \int_{\Rnum^n}\phi^jf^i\rho dy,\\
\int_{\Rnum^n}\tilde\phi^if^j\rho dy = -(\tilde\phi^i,L_0\phi^j)_\rho
= -(L_0\tilde\phi^i,\phi^j)_\rho = \int_{\Rnum^n}\phi^j\tilde{f}^i\rho dy.
\end{split}
\end{equation}
These two symmetric relations, together with Assumption \ref{a1} and integration by parts, show that
\begin{equation*}
\begin{split}
\bar{A}^{m+1,j}(x,t) &= \int_{\Rnum^n}[\bar{a}^{m+1,j}
+2\phi^{m+1}(f^j-\tfrac{1}{2}\partial_{y_k}h^{jk}-\tfrac{1}{2}h^{jk}\partial_{y_k}\log\rho)
+\partial_{y_k}\phi^j\bar{h}^{m+1,k}]\rho dy\\
&= \int_{\Rnum^n}[\bar{a}^{m+1,j}+\phi^{m+1}(f^j-\tilde{f}^j)
+\partial_{y_k}\phi^j\bar{h}^{m+1,k}]\rho dy.
\end{split}
\end{equation*}
Moreover, for any $1\leq i\leq m$, using Assumption \ref{a1}, integration by parts, and \eqref{symmetric} again yields
\begin{equation*}
\begin{split}
&\;w^i+\tilde{w}^i-(\nabla_x\cdot A)^i\\
=&\;\int_{\Rnum^n}(b^i+\partial_{x_j}\phi^if^j+\partial_{y_j}\phi^ig^j+\partial_{x_jy_k}\phi^ih^{jk}
+\tilde{b}^i+\partial_{x_j}\tilde\phi^i\tilde{f}^j+\partial_{y_j}\tilde\phi^i\tilde{g}^j
+\partial_{x_jy_k}\tilde\phi^ih^{jk})\rho\\
&\;-\partial_{x_j}\int_{\Rnum^n}(a^{ij}+\phi^if^j+\phi^jf^i+\partial_{y_k}\phi^ih^{jk}
+\partial_{y_k}\phi^jh^{ik})\rho dy\\
=&\;\int_{\Rnum^n}(b^i+\tilde{b}^i-\partial_{x_j}a^{ij}-a^{ij}\partial_{x_j}\log\rho)\rho dy
-\int_{\Rnum^n}\phi^i[\partial_{x_j}(f^j\rho)+\partial_{y_j}(g^j\rho)-\partial_{x_jy_k}(h^{jk}\rho)]dy\\
&\;+\int_{\Rnum^n}\partial_{x_j}\tilde\phi^i(\tilde{f}^j-\partial_{y_k}h^{jk}-h^{jk}\partial_{y_k}\log\rho)
\rho dy-\int_{\Rnum^n}\tilde\phi^i\partial_{y_j}(\tilde{g}^j\rho)dy\\
&\;-\partial_{x_j}\int_{\Rnum^n}\phi^j(f^i-\partial_{y_k}h^{ik}-h^{ik}\partial_{y_k}\log\rho)\rho dy\\
=&\;\int_{\Rnum^n}\bar{a}^{m+1,i}\rho dy
+\int_{\Rnum^n}\phi^i[\partial_{y_j}(\tilde{g}^j\rho)-\partial_{x_j}(f^j\rho)]dy\\
&\;-\int_{\Rnum^n}\phi^i\partial_{y_k}[(g^k+\tilde{g}^k-\partial_{x_j}h^{jk}-h^{jk}\partial_{x_j}\log\rho)\rho]dy\\
&\;-\partial_{x_j}\int_{\Rnum^n}\tilde\phi^if^j\rho dy
-\int_{\Rnum^n}\tilde\phi^i[\partial_{y_j}(\tilde{g}^j\rho)-\partial_{x_j}(f^j\rho)]dy
+\partial_{x_j}\int_{\Rnum^n}\phi^j\tilde{f}^i\rho dy\\
=&\;\int_{\Rnum^n}\bar{a}^{m+1,i}\rho dy+\int_{\Rnum^n}\phi^if^{m+1}\rho dy
+\int_{\Rnum^n}\partial_{y_k}\phi^i\bar{h}^{m+1,k}\rho dy-\int_{\Rnum^n}\tilde\phi^if^{m+1}\rho dy\\
=&\;\int_{\Rnum^n}[\bar{a}^{m+1,i}+\phi^{m+1}(f^i-\tilde{f}^i)+\partial_{y_k}\phi^i\bar{h}^{m+1,k}]\rho dy
= \bar{A}^{m+1,i}(x,t).
\end{split}
\end{equation*}
This completes the proof.
\end{proof}

\begin{lemma}\label{w3}
\begin{equation*}
\nabla_x\cdot(\tilde{w}-\tfrac{1}{2}\nabla_x\cdot A)
= \nabla_x\cdot\int_{\Rnum^n}(\tilde{b}-\tfrac{1}{2}\nabla_x\cdot a-\tfrac{1}{2}a\nabla\log\rho
-\phi^{m+1}\tilde{f})\rho dy.
\end{equation*}
\end{lemma}

\begin{proof}
For any $1\leq i\leq m$, using Assumption \ref{a1}, integration by parts, and \eqref{symmetric} yields
\begin{equation*}
\begin{split}
&\;\tilde{w}^i-\tfrac{1}{2}(\nabla_x\cdot A)^i\\
=&\;\int_{\Rnum^n}(\tilde{b}^i+\partial_{x_j}\tilde\phi^i\tilde{f}^j+\partial_{y_j}\tilde\phi^i\tilde{g}^j
+\partial_{x_jy_k}\tilde\phi^ih^{jk})\rho dy\\
&\;-\tfrac{1}{2}\partial_{x_j}\int_{\Rnum^n}(a^{ij}+\phi^if^j+\phi^jf^i+\partial_{y_k}\phi^ih^{jk}
+\partial_{y_k}\phi^jh^{ik})\rho dy\\
=&\;\int_{\Rnum^n}(\tilde{b}^i-\tfrac{1}{2}\partial_{x_j}a^{ij}-\tfrac{1}{2}a^{ij}\partial_{x_j}\log\rho)\rho dy\\
&\;+\int_{\Rnum^n}\partial_{x_j}\tilde\phi^i(\tilde{f}^j-\partial_{y_k}h^{jk}-h^{jk}\partial_{y_k}\log\rho)\rho dy-\int_{\Rnum^n}\tilde\phi^i\partial_{y_j}(\tilde{g}^j\rho)dy\\
&\;-\tfrac{1}{2}\partial_{x_j}\int_{\Rnum^n}\phi^j(f^i-\partial_{y_k}h^{ik}-h^{ik}\partial_{y_k}\log\rho)\rho dy
-\tfrac{1}{2}\partial_{x_j}\int_{\Rnum^n}\phi^i(f^j-\partial_{y_k}h^{jk}-h^{jk}\partial_{y_k}\log\rho)\rho dy\\
=&\;\int_{\Rnum^n}(\tilde{b}^i-\tfrac{1}{2}\partial_{x_j}a^{ij}-\tfrac{1}{2}a^{ij}\partial_{x_j}\log\rho)\rho dy-\partial_{x_j}\int_{\Rnum^n}\tilde\phi^if^j\rho dy
-\int_{\Rnum^n}\tilde\phi^i[\partial_{y_j}(\tilde{g}^j\rho)-\partial_{x_j}(f^j\rho)]dy\\
&\;+\tfrac{1}{2}\partial_{x_j}\int_{\Rnum^n}\phi^j\tilde{f}^i\rho dy
+\tfrac{1}{2}\partial_{x_j}\int_{\Rnum^n}\phi^i\tilde{f}^j\rho dy\\
=&\;\int_{\Rnum^n}(\tilde{b}^i-\tfrac{1}{2}\partial_{x_j}a^{ij}-\tfrac{1}{2}a^{ij}\partial_{x_j}\log\rho)\rho dy-\int_{\Rnum^n}\tilde\phi^if^{m+1}\rho dy
+\tfrac{1}{2}\partial_{x_j}\int_{\Rnum^n}(\phi^i\tilde{f}^j-\phi^j\tilde{f}^i)\rho dy\\
=&\;\int_{\Rnum^n}(\tilde{b}^i-\tfrac{1}{2}\partial_{x_j}a^{ij}-\tfrac{1}{2}a^{ij}\partial_{x_j}\log\rho
-\phi^{m+1}\tilde{f}^i)\rho dy
+\tfrac{1}{2}\partial_{x_j}\int_{\Rnum^n}(\phi^i\tilde{f}^j-\phi^j\tilde{f}^i)\rho dy.
\end{split}
\end{equation*}
Therefore, we obtain
\begin{equation*}
\nabla_x\cdot(\tilde{w}-\tfrac{1}{2}\nabla_x\cdot A)
= \nabla_x\cdot\int_{\Rnum^n}(\tilde{b}-\tfrac{1}{2}\nabla_x\cdot a-\tfrac{1}{2}a\nabla\log\rho
-\phi^{m+1}\tilde{f})\rho dy,
\end{equation*}
where we have used the fact that
\begin{equation*}
\partial_{x_ix_j}(\phi^i\tilde{f}^j-\phi^j\tilde{f}^i) = 0.
\end{equation*}
This completes the proof.
\end{proof}

We are now in a position to prove Theorem \ref{backward}.

\begin{proof}[Proof of Theorem \ref{backward}]
As discussed earlier in this section, we have proved that $G^\epsilon_t$ and $G^{(1)}_t$ are well defined for each $t$ and $(X^\epsilon_t,G^\epsilon_t) \Rightarrow (X_t,G_t)$ in $C([0,T],\Rnum^{m+1})$. We shall next prove that $G^{(2)}_t$ is a martingale. From \eqref{G1} and \eqref{G}, we obtain
\begin{equation*}
\begin{split}
G^2_T =&\; G_T-G^{(1)}_T\\
=&\; \int_0^T[(\bar{A}^{1/2})^{m+1,\cdot}-(w+\tilde{w}-\nabla_x\cdot A)'A^{-1}(\bar{A}^{1/2})^{x,\cdot}](X_s,s)dB_s\\
&\;+\int_0^T\left[\bar{w}^{m+1}-\frac{1}{2}(w+\tilde{w}-\nabla_x\cdot A)A^{-1}(w+\tilde{w}-\nabla_x\cdot A)-\nabla_x\cdot(\tilde{w}-\tfrac{1}{2}\nabla_x\cdot A)\right](X_s,s)ds.
\end{split}
\end{equation*}
Straightforward calculations show that
\begin{equation*}
\begin{split}
&\;[(\bar{A}^{1/2})^{m+1,\cdot}-(w+\tilde{w}-\nabla_x\cdot A)'A^{-1}(\bar{A}^{1/2})^{x,\cdot}]
[(\bar{A}^{1/2})^{m+1,\cdot}-(w+\tilde{w}-\nabla_x\cdot A)'A^{-1}(\bar{A}^{1/2})^{x,\cdot}]' \\
=&\; \bar{A}^{m+1,m+1}-2\bar{A}^{m+1,x}A^{-1}(w+\tilde{w}-\nabla_x\cdot A)+(w+\tilde{w}-\nabla_x\cdot A)'A^{-1}(w+\tilde{w}-\nabla_x\cdot A).
\end{split}
\end{equation*}
Combining Lemmas \ref{w1} and \ref{w3}, we obtain
\begin{equation*}
\bar{A}^{m+1,m+1} = 2[\bar{w}^{m+1}-\nabla_x\cdot(\tilde{w}-\tfrac{1}{2}\nabla_x\cdot A)].
\end{equation*}
This fact, together with Lemma \ref{w2}, shows that
\begin{equation*}
\begin{split}
&\;[(\bar{A}^{1/2})^{m+1,\cdot}-(w+\tilde{w}-\nabla_x\cdot A)'A^{-1}(\bar{A}^{1/2})^{x,\cdot}]
[(\bar{A}^{1/2})^{m+1,\cdot}-(w+\tilde{w}-\nabla_x\cdot A)'A^{-1}(\bar{A}^{1/2})^{x,\cdot}]' \\
=&\; 2[\bar{w}^{m+1}-\nabla_x\cdot(\tilde{w}-\tfrac{1}{2}\nabla_x\cdot A)]
-(w+\tilde{w}-\nabla_x\cdot A)'A^{-1}(w+\tilde{w}-\nabla_x\cdot A).
\end{split}
\end{equation*}
This shows that $e^{-G^{(2)}_t}$ is an exponential martingale.
\end{proof}

\section*{Acknowledgments}
The authors acknowledge George G. Yin and Wenqing Hu for valuable comments and stimulating discussions. H. Ge is supported by NSFC (No. 11971037 and 11622101). C. Jia acknowleges support from startup funds provided by the Beijing Computational Science Research Center.

\setlength{\bibsep}{5pt}
\small\bibliographystyle{nature}

\end{document}